%% file: x-y-SymmetryV2.tex
\documentclass[a4paper,12pt,reqno]{amsart}

\usepackage{a4}
\usepackage{verbatim}
\usepackage{todonotes}
\usepackage{color}

\usepackage{amssymb}
\usepackage{color}
\usepackage{graphicx}
\usepackage{floatflt}
\usepackage{float} 		

\usepackage{amsthm}
\usepackage{amscd}
\usepackage{hyperref}

\newtheorem{definition}{Definition}
\newtheorem{example}[definition]{Example}

\newtheorem{theorem}[definition]{Theorem}
\newtheorem{lemma}[definition]{Lemma}
\newtheorem{remark}[definition]{Remark}

\newtheorem{corollary}[definition]{Corollary}

\newtheorem{proposition}[definition]{Proposition}

\newtheorem{assumption}[definition]{Assumption}

\newcommand{\widesim}[2][1.5]{
	\mathrel{\overset{#2}{\scalebox{#1}[1]{$\sim$}}}
}

\def\XXint#1#2#3{{\setbox0=\hbox{$#1{#2#3}{\int}$}
		\vcenter{\hbox{$#2#3$}}\kern-.5\wd0}}

\makeatletter
\@addtoreset{definition}{section}
\@addtoreset{equation}{section}
\makeatother

\DeclareMathOperator{\Res}{Res}

\allowdisplaybreaks[4]

\begin{document}

\title[On the $x$-$y$ Symmetry of Correlators in Topological Recursion via Loop Insertion Operator]{On the $x$-$y$ Symmetry of Correlators in Topological Recursion via Loop Insertion Operator}

\author{Alexander Hock}

\address{Mathematical Institute, University of Oxford, Andrew Wiles Building, Woodstock Road,
	OX2 6GG, Oxford, UK \\
{\itshape E-mail address:} \normalfont  
\texttt{alexander.hock@maths.ox.ac.uk}}

\maketitle

\markboth{\hfill\textsc\shortauthors}{\textsc{On the $x$-$y$ Symmetry in Topological Recursion via Loop Insertion Operator}\hfill}


\begin{abstract}
	  Topological Recursion generates a family of symmetric differential forms (correlators) from some initial data $(\Sigma,x,y,B)$. We give a functional relation between the correlators of genus $g=0$ generated by the initial data $(\Sigma,x,y,B)$ and by the initial data $(\Sigma,y,x,B)$, where $x$ and $y$ are interchanged. The functional relation is derived with the loop insertion operator by computing a functional relation for some intermediate correlators. Additionally, we show that our result is equivalent to the recent result of \cite{Borot:2021thu} in case of $g=0$. Consequently, we are providing a simplified functional relation between generating series of higher order free cumulants and moments in higher order free probability.
\end{abstract}

\section{Introduction}
Topological Recursion (TR), invented in 2007 by Chekhov, Eynard and Orantin \cite{Chekhov:2006vd,Eynard:2007kz},  aroused more and more interest in last few years. It is a universal recursive procedure to generate a family of correlators. Depending on the initial data, the \textit{spectral curve}, which gives the starting point of the recursion, TR has various applications in different areas of mathematics and mathematical physics, we refer to \cite{Eynard:2014zxa} for a short overview. 

Important invariants of TR are the so-called symplectic invariants $\mathcal{F}^{(g)}$, which are conjectured to be invariant under all symplectic transformations of the spectral curve. However, there is one striking transformation, the $x$-$y$ interchange, where the invariance is not proved in its full generality. In this article, we are interested on the implication of the $x$-$y$ interchange on the correlators themselves, which are not invariant at all. We will give relations between the planar correlators under $x$-$y$ interchange using the loop insertion operator, and confirm them by some examples.

There are two established examples, which should be kept in mind, where the $x$-$y$ transformation generates correlators of combinatorial meaning:
\begin{itemize}
	\item The hermitian 2-matrix model \cite{Chekhov:2006vd,Eynard:2007nq}, where the $x$-$y$ transformation interchanges the colours or equivalently the potentials $V_1$ and $V_2$.
	\item The Hermitian 1-matrix model, where the $x$-$y$ transformation relates ordinary maps to fully simple maps \cite{Borot:2021eif}. Consequently, this has a further application to free probability (see \cite{Garcia-Failde:2019iuf} for some discussions).
\end{itemize}

To be more precise, we are now defining the correlators of
TR. Let $\omega^{(g)}_{n,0}$ (the so-called correlators) be a family of symmetric
meromorphic differentials on $n$ products of Riemann surfaces $\Sigma$.
These $\omega^{(g)}_{n,0}$ are labeled by the genus $g$ and the number $n$
of marked points of a compact complex curve. These objects occur as correlators
of algebraic curves $E(x,y)=0$, understood in parametric
representation $x(z)$ and $y(z)$. For simplicity, we will assume that the algebraic curve is of genus zero, which means that a rational parametrisation $x(z)$ and $y(z)$ exists.

From the initial  data, the \textit{spectral curve} $(\Sigma,x,y,B)$, consisting of a ramified covering
$x: \Sigma \to \Sigma_0$ of Riemann surfaces, 
a differential 1-form $\omega^{(0)}_{1,0}(z)=y(z)dx(z)$  and the \emph{Bergman kernel}
$\omega^{(0)}_{2,0}(z_1,z_2)=B(z_1,z_2)=\frac{dz_1\,dz_2}{(z_1-z_2)^2}$,
TR constructs the meromorphic differentials $\omega^{(g)}_{n+1,0}(z_1,...,z_n,z)$ with
$2g+n\geq 2$ via the following universal formula
(in which we abbreviate $I=\{z_1,...,z_n\}$):
\begin{align}
	\label{eq:trx}
	& \omega^{(g)}_{n+1,0}(I,z)
	\\
	& =\sum_{\alpha_i}
	\Res\displaylimits_{q\to \alpha_i}
	K_i(z,q)\bigg(
	\omega^{(g-1)}_{n+2,0}(I, q,\sigma_i(q))
	+\hspace*{-1cm} \sum_{\substack{g_1+g_2=g\\ I_1\uplus I_2=I\\
			(g_1,I_1)\neq (0,\emptyset)\neq (g_2,I_2)}}
	\hspace*{-1.1cm} \omega^{(g_1)}_{|I_1|+1,0}(I_1,q)
	\omega^{(g_2)}_{|I_2|+1,0}(I_2,\sigma_i(q))\!\bigg)\;.
	\nonumber
\end{align}
This construction proceed recursively 
in the negative Euler characteristic $-\chi=2g+n-2$. Further, we need to define:
\begin{itemize}
	 \item The sum over the \textit{ramification points} $\alpha_i$ of the ramified
	covering  $x:\Sigma\to \Sigma_0$, defined via $dx(\alpha_i)=0$.
	\item The \textit{local Galois involution} $\sigma_i\neq \mathrm{id}$
	defined via $x(q)=x(\sigma_i(q))$ near $\alpha_i$ with the fixed
	point $\alpha_i$. 
	\item The \textit{recursion kernel} $K_i(z,q)
	=\frac{\frac{1}{2}\int^{q'=q}_{q'=\sigma_i(q)}
		B(z,q')}{\omega^{(0)}_{1,0}(q)-\omega^{(0)}_{1,0}(\sigma_i(q))}$  constructed
	from the initial data. 
\end{itemize}
We also assume that $y$ is regular at the ramification points of $x$ and vice versa, and they have no coinciding ramification points.

Applying the $x$-$y$ transformation, we are interchanging the role of $x$ and $y$ and define by the same recursion (with $x$ and $y$ interchanged) $\omega^{(g)}_{0,m}$ a family of symmetric
meromorphic differentials on $m$ products of Riemann surfaces $\Sigma$. The initial data, the spectral curve $(\Sigma,y,x,B)$, consists of $\omega^{(0)}_{0,1}(z)=x(z)dy(z)$, $\omega^{(0)}_{0,2}(z_1,z_2)=B(z_1,z_2)=\frac{dz_1\,dz_2}{(z_1-z_2)^2}$. All $\omega^{(g')}_{0,m'}$ are used in the recursive formula \eqref{eq:trx} with residue at the ramification points of $y$, which we call $\beta_{i}$, i.e. $dy(\beta_i)=0$. 

For negative Euler characteristic, it is known that all $\omega^{(g)}_{n,0}$ have poles only at $\alpha_i$ the ramification points of $x$ and all $\omega^{(g)}_{0,m}$ have poles only at $\beta_{i}$ the ramification points of $y$. This article proves the following functional relations between $\omega^{(0)}_{0,m}$ and $\omega^{(0)}_{n,0}$
\begin{theorem}\label{thm:first}
	Consider Assumption \ref{Ass}.
	Let $\omega^{(g)}_{n,0}$ be generated by TR with the spectral curve $(\Sigma,x,y,B)$ and $\omega^{(g)}_{0,m}$ generated by TR with the spectral curve $(\Sigma,y,x,B)$, and let $W^{(g)}_{n,0}(x_1(z_1),...,x_n(z_n)):=\frac{\omega^{(g)}_{n,0}(z_1,...,z_n)}{dx_1(z_1)...dx_n(z_n)}$ and $W^{(g)}_{0,m}(y_1(z_1),...,y_m(z_m)):=\frac{\omega^{(g)}_{0,m}(z_1,...,z_n)}{dy_1(z_1)...dy_m(z_m)}$. Then
	the functional relation of the correlator $W^{(0)}_{0,m}$ reads
	\begin{align*}
		&W^{(0)}_{0,m}(y(z_1),...,y(z_m))\\
		=&\sum_{T\in \mathcal{G}_{0,m}}\frac{1}{\mathrm{Aut}(T)}\prod_{j=1}^m\bigg(-\frac{d}{d y(z_j)}\bigg)^{r_j(T)-1}\bigg(\prod_{k=1}^m\bigg(-\frac{dx(z_k)}{dy(z_k)}\bigg)\prod_{(\emptyset,I)\in \mathcal{I}(T)}W^{(0)}_{|I|,0}\big(x(z_I))\big)\bigg),
	\end{align*}
	where the set of trees $\mathcal{G}_{n,m}$ and the associated set $\mathcal{I}(T)$ are defined in Definition \ref{def:tree}. Furthermore, we need $r_j(T)$ as the valence of the $j^{\text{th}}$ $\circ$-vertex of $T$, and the abbreviation $x(z_I)=\{x(z_{i_1}),...,x(z_{i_{|J|}})\}$ for some set $I=\{i_1,...,i_{|I|}\}$.
\end{theorem} 

This is achieved by making sense of $\omega^{(g)}_{n,m}$ (or better $W^{(g)}_{n,m}$) for any $n,m$. Theorem \ref{thm:main} gives the more general functional relation between $W^{(0)}_{n,m}$  and $W^{(0)}_{n',0}$. The technical tool in all computations is the \textit{loop insertion operator}. The derivations carried out in this article are motivated by recent results in the Quartic Kontsevich model \cite{Branahl:2020yru,Hock:2021tbl,Branahl:2021uea}. 
It is more than worth to notice that actually a relation between $\omega^{(g)}_{n,0}$ and $\omega^{(g)}_{0,m}$ is known for any genus $g$ \cite{Borot:2021thu} (after simplification equivalent to Theorem \ref{thm:first} for $g=0$). Their proof is built on a series of papers by one of the others together with Bychkov, Dunin-Barkowski and Kazarian \cite{Bychkov:2020ujd,Bychkov:2020yzy,Bychkov:2021hfh}.
It includes a wide class of spectral curves related to double Hurwitz numbers. However, our Theorem \ref{thm:first} gives a more canonical representaion in terms of derivatives wrt to $y$ instead of $\ln y$. The represention in terms of $\ln y$ suggests first of all a more complicated pole structure, but it was proved in \cite{Bychkov:2020yzy} that all other poles except for the ramification points cancel.  In other words, Theorem \ref{thm:first} is more canonical in the sense that it becomes obvious that all $W^{(g)}_{0,m}$ with $m>2$ have only poles at the ramification points of $y$ (that is $dy=0)$. We show in Sec. \ref{Sec:eq} how in the Theorem of \cite{Borot:2021thu} the additional poles at the zeros of $y$ disappear.

We emphasise that inserting the spectral curve related to the duality between ordinary and fully simple maps \cite{Borot:2021eif}, the functional relation for the $R$-transform in higher order free probability \cite{Voiculescu1986AdditionOC,Collins2006SecondOF} is given by a much simpler formula in Corollary \ref{cor:free} than in \cite{Borot:2021thu}.

\subsection*{Acknowledgements}

I thank Sergey Shadrin, Gaetan Borot and Johannes Branahl for discussions. This work was supported through
the Walter-Benjamin fellowship\footnote{``Funded by
	the Deutsche Forschungsgemeinschaft (DFG, German Research
	Foundation) -- Project-ID 465029630}.

\section{The 2-Matrix Model as Prime Example}
To get a feeling for the later computations, we will discuss the prime example the 2-matrix model (see \cite{Eynard:2002kg,Chekhov:2006vd} for more details). Let $H_N$ be the space of $N\times N$ hermitian matrices, then the partition function of the 2-matrix model is defined by
\begin{align}\label{Z2MM}
	\mathcal{Z}^{2MM}:=\int_{H_N}dM_1\,dM_2\,e^{-N\mathrm{Tr}[V_1(M_1)+V_2(M_2)-M_1M_2]}=e^{-N^2 \mathcal{F}},
\end{align} 
where $V_i(x)=\sum_{k=1}^{d_i+1}\frac{t^{(i)}_{k}}{k}x^k$ for $i=1,2$ are polynomials of degree $d_i+1$ and $\mathcal{F}$ is called the free energy. The partition function \eqref{Z2MM} has to be understood as a formal matrix model in the sense that the expansion of the integrand (except for the Gaussian parts) is interchanged with the integral. With the probability weight $d\mu_{2MM}:=\frac{1}{\mathcal{Z}^{2MM}}dM_1\,dM_2\,e^{-N\mathrm{Tr}[V_1(M_1)+V_2(M_2)-M_1M_2]}$, we define the following correlators by the cumulants (connected components) of the resolvents 
\begin{align}\label{2MMW}
	\hat{W}_{n,m}(x_1,...,x_n|y_1,...,y_m):=&N^{n+m-2}\bigg\langle \prod_{i=1}^{n}\frac{1}{x_i-M_1}\prod_{j=1}^{m}\frac{1}{y_j-M_2}\bigg\rangle_{c}\\\nonumber
	=&-\frac{\partial}{\partial V_1(x_1)}...\frac{\partial}{\partial V_1(x_n)}\frac{\partial}{\partial V_2(y_1)}...\frac{\partial}{\partial V_2(y_m)}\mathcal{F}\\\nonumber
	&+\frac{\delta_{n,1}\delta_{m,0}}{x_1}+\frac{\delta_{n,0}\delta_{m,1}}{y_1},
\end{align}
where the formal loop insertion operators are given by
\begin{align}\label{insertion2MM}
	\frac{\partial}{\partial V_1(x)}=\sum_{i=1}^\infty\frac{i}{x^{i+1}}\frac{\partial}{\partial t^{(1)}_i},\qquad \frac{\partial}{\partial V_2(y)}=\sum_{j=1}^\infty\frac{j}{y^{j+1}}\frac{\partial}{\partial t^{(2)}_j}.
\end{align}
To reduce later formulas, we introduce a shift for the first topologies:
\begin{align*}
	W_{n,m}(x_1,...,x_n|y_1,...,&y_m):=\hat{W}_{n,m}(x_1,...,x_n|y_1,...,y_m)\\
	&-\delta_{n,1}\delta_{m,0}V'_1(x_1)-\delta_{n,0}\delta_{m,1}V'_2(y_1)+\frac{\delta_{n,2}\delta_{m,0}}{(x_1-x_2)^2}+\frac{\delta_{n,0}\delta_{m,2}}{(y_1-y_2)^2},
\end{align*}
The correlators $W_{n,m}$ and the free energy are known to have a formal genus expansion
\begin{align*}
	W_{n,m}=\sum_{g=0}^\infty N^{-2g}W^{(g)}_{n,m},\qquad \mathcal{F}=\sum_{g=0}^\infty N^{-2g}\mathcal{F}^{(g)}.
\end{align*}
By definition, the loop insertion operator acts as follows
\begin{align}\label{Wn+12MM}
	W^{(g)}_{n+1,m}(x_1,...,x_n,x|y_1,...,y_m)=\frac{\partial}{\partial V_1(x)}W^{(g)}_{n,m}(x_1,...,x_n,x|y_1,...,y_m)\\\label{Wm+12MM}
	W^{(g)}_{n,m+1}(x_1,...,x_n|y_1,...,y_m,y)=\frac{\partial}{\partial V_2(y)}W^{(g)}_{n,m}(x_1,...,x_n,x|y_1,...,y_m).
\end{align} 
This action is independent of $x_i$ and $y_j$, so it is an action on $W^{(g)}_{n,m}$ itself and not on its variables.

The amazing discovery of \cite{Chekhov:2006vd} was that all correlators $W^{(g)}_{n,0}$ are computed by TR:
\begin{theorem}[\cite{Chekhov:2006vd}]\label{thm:2MM}
	The correlators are given by
	\begin{align*}
		W^{(g)}_{n,0}(x_1(z_1),....,x_n(z_n))=\frac{\omega^{(g)}_{n,0}(z_1,...,z_n)}{dx_1(z_1)...dx_n(z_n)},
	\end{align*}
	where $\omega^{(g)}_{n,0}(z_1,...,z_n)$ is given by TR \eqref{eq:trx} and $x(z),y(z)$ are a solution of the master loop equation \cite[eq. (2-17)]{Chekhov:2006vd} $E(x(z),y(z))=0$, which is an algebraic equation. 
	
	For the genus zero assumption on $E(x(z),y(z))=0$, the solution has a rational parametrisation and is explicit \cite{Eynard:2002kg}
	\begin{align*}
		x(z)=\gamma z+\sum_{k=0}^{d_2}\frac{\tau^{(1)}_k}{z^k},\qquad y(z)=\frac{\gamma}{z} +\sum_{k=0}^{d_1}\tau^{(2)}_k z^k,
	\end{align*}
where $\tau^{(i)}_k$ satisfies
\begin{align*}
	y(z)-V'_1(x(z))&\widesim[2]{s\to \infty} -\frac{\gamma}{s}+\mathcal{O}(s^{-2})\\
	x(z)-V'_2(y(z))&\widesim[2]{s\to 0} -\frac{s}{\gamma}+\mathcal{O}(s^{2}).
\end{align*}
\end{theorem}
By interchanging $V_1 \leftrightarrow V_2$, which is nothing than interchanging $x\leftrightarrow y$, we get
\begin{align*}
	W^{(g)}_{0,m}(y_1(z_1),....,y_m(z_m))=\frac{\omega^{(g)}_{0,m}(z_1,...,z_m)}{dy_1(z_1)...dy_m(z_m)},
\end{align*}
where $\omega^{(g)}_{0,m}$ is given by \eqref{eq:trx} after interchanging $x$ and $y$.

From Theorem \ref{thm:2MM}, the disc amplitude is
\begin{align}\label{1Pz2MM}
	\hat{W}^{(0)}_{1,0}(x(z))=y(z)+V'_1(x(z)),\qquad \hat{W}^{(0)}_{0,1}(y(z))=x(z)+V'_2(y(z)).
\end{align}
The action of the loop insertion operator \eqref{insertion2MM} on the disc amplitude generates the cylinder amplitude, however it acts on a function with variable $x$ (or $y$, respectively), where $x$ (or $y$) is kept fixed.  
\begin{figure}[h]
	\scalebox{0.8}{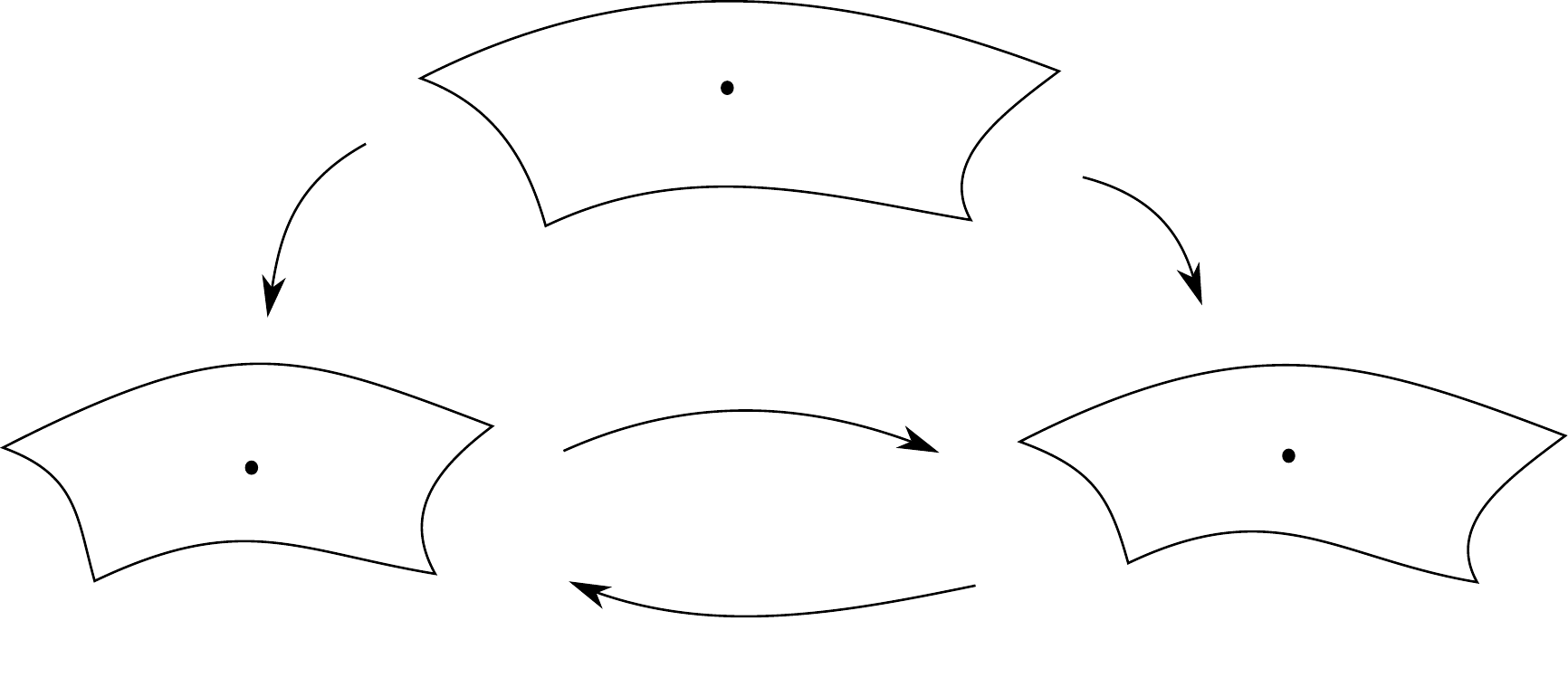}
	\caption{Coverings $x(z)$ and $y(z)$ of Riemann surfaces. The maps $X(y)$ and $Y(x)$ are functional inverses to each other, and locally multivalued functions depending on the number of branches of $x(z)$ and $y(z)$.}
	\label{fig:xymap}
\end{figure}
The rhs of \eqref{1Pz2MM} is written in terms of $x$ (or $y$, respectively) by
\begin{align}\label{1Px2MM}
	\hat{W}^{(0)}_{1,0}(x)=Y(x)+V'_1(x),\qquad \hat{W}^{(0)}_{0,1}(y)=X(y)+V'_2(y),
\end{align}
where $Y(x)$ (or $X(y)$) are in general multivalued functions depending on the number of branches of $x(z)$ (or $y(z)$), see \cite[Sec. 4.4]{Eynard:2002kg} for further details. We have the functional inversion
\begin{align}
	x=X(y)\qquad \iff \qquad y=Y(x).
\end{align}

Now, the loop insertion operator can be applied since the rhs is expressed in terms of the variable $x$ (or $y$, respectively). Due to \eqref{Wn+12MM}, we have
\begin{align*}
	\hat{W}^{(0)}_{2,0}(x_1,x_2)=&\frac{\partial}{\partial V_1(x_2)}\hat{W}^{(0)}_{1,0}(x_1)=\frac{\partial}{\partial V_1(x_2)}\bigg(Y(x_1)+V'_1(x_1)\bigg)\\
	=&\frac{1}{x'(z_1(x_1))x'(z_2(x_2))(z_1(x_1)-z_2(x_2))^2}-\frac{1}{(x_1-x_2)^2},
\end{align*} 
where $\frac{\partial V'_1(x_1)}{\partial V_1(x_2)}=-\frac{1}{(x_1-x_2)^2}$ by applying the geometric series for $|\frac{x_2}{x_1}|<1$. This explains why \eqref{Wn+12MM} and \eqref{Wm+12MM} holds for $W$ (and not only for $\hat{W}$). We write $z_i(x_i)$ as some local inverse of $x_i(z_i)$, where the chosen branch does not matter, however there is some canonical ''physical'' branch. Consequently, the action of the loop insertion operator on $Y(x)$ is
\begin{align*}
	\frac{\partial}{\partial V_1(x_2)}Y(x_1)=\frac{\omega^{(0)}_{2,0}(z_1,z_2)}{dx_1(z_1)dx_2(z_2)}\bigg\vert_{z_i=z_i(x_i)}=W^{(0)}_{2,0}(x_1,x_2).
\end{align*}
Trivially, we can do the same computation for $X(y)$. 

As an interesting example, we get the correlator $W^{(0)}_{1,1}(x|y)$ by the previous rules of the loop insertion operator for free, which is consistent with \cite{Daul:1993bg}:
\begin{lemma}\label{lem:W11}
	The cylinder correlator with two boundaries of different colours is 
	\begin{align*}
		W^{(0)}_{1,1}(x|y)=-\frac{1}{x'(z_1(x))y'(z_2(y))(z_1(x)-z_2(y))^2}=-\frac{dX(y)}{dy}W^{(0)}_{2,0}(x,X(y)).
	\end{align*}
\begin{proof}
	Making use of the previously discussed properties of the loop insertion operator and the chain rule, we get
	\begin{align*}
		0=&\frac{\partial}{\partial V_1(x(z_1))}y=\frac{\partial}{\partial V_1(x(z_1))}Y(X(y))\\
		=&\frac{1}{x'(z_2(X(y)))x'(z_1(x))(z_2(X(y))-z_1(x))^2}+Y'(X(y))\frac{\partial}{\partial V_1(x)}X(y)\\
		=&\frac{1}{x'(z_2(X(y)))x'(z_1(x))(z_1(X(y))-z_1(x))^2}+Y'(X(y))W^{(0)}_{1,1}(x|y),
	\end{align*}
	where we use the notation that $z(y)$ is an inverse of $y(z)$ and $z(x)$ an inverse of $x(z)$. The chain rule $\frac{d Y(x)}{d x}\vert_{x=x(z)}\cdot\frac{d x(z)}{d z}=\frac{d Y(x(z))}{d z}=\frac{d y(z)}{d z}$ finishes the proof.
\end{proof}
\end{lemma}
This computation shows that the loop insertion operator has to act on functions represented by variables $x_i,y_j$ and not $z_i$. This is due to the fact, that $\frac{\partial}{\partial V_1(x)}x_i=\frac{\partial}{\partial V_1(x)}y_j=0$, whereas
$\frac{\partial}{\partial V_1(x)}z_i(x_i)\neq 0$, because $z_i(x_i)$ is some inverse of $x_i(z_i)$ depending on $t^{(1)}_k$. We should ensure that all variables are either on the $x$- or on the $y$-plane, see Fig. \ref{fig:xymap}.

We finish this section by recapping known examples for the pair of pants topology (see for instance \cite[Sec. A.1.2]{Chekhov:2006vd} and \cite[eq. 5.14]{Borot:2021eif}):
\begin{proposition}\label{prop:W32MM}
	For the pair of pants, we have the following functional relations:
	\begin{align}\label{W122MM}
		W^{(0)}_{2,1}(x_1,x_3|Y(x_2))=&-\frac{dx_2}{dY(x_2)}W^{(0)}_{3,0}(x_1,x_2,x_3)\\\nonumber
		&-\frac{d}{dY(x_2)}\bigg(\frac{dx_2}{dY(x_2)}W^{(0)}_{2,0}(x_2,x_3)W^{(0)}_{2,0}(x_1,x_2)\bigg)
	\end{align}
	\begin{align}
		W^{(0)}_{0,3}(Y(x_1),Y(x_2)&,Y(x_3))=-\frac{dx_1}{dY(x_1)}\frac{dx_2}{dY(x_2)}\frac{dx_3}{dY(x_3)}W^{(0)}_{3,0}(x_1,x_2,x_3)\\\nonumber
		&+\frac{d}{dY(x_1)}\bigg(\frac{dx_1}{dY(x_1)}\frac{dx_2}{dY(x_2)}\frac{dx_3}{dY(x_3)}W^{(0)}_{2,0}(x_1,x_2)W^{(0)}_{2,0}(x_1,x_3)\bigg)\\\nonumber
		&+\frac{d}{dY(x_2)}\bigg(\frac{dx_1}{dY(x_1)}\frac{dx_2}{dY(x_2)}\frac{dx_3}{dY(x_3)}W^{(0)}_{2,0}(x_1,x_2)W^{(0)}_{2,0}(x_2,x_3)\bigg)\\\nonumber
		&+\frac{d}{dY(x_3)}\bigg(\frac{dx_1}{dY(x_1)}\frac{dx_2}{dY(x_2)}\frac{dx_3}{dY(x_3)}W^{(0)}_{2,0}(x_1,x_3)W^{(0)}_{2,0}(x_2,x_3)\bigg).
	\end{align}
	These functional relations are also true on the $z$-plane by writing $Y(x_i)=y(z_i)$ and $x_i=x(z_i)$, or on the $y$-plane by writing $Y(x_i)=y_i$ and $x_i=X(y_i)$.
	\begin{proof}
		Starting with Lemma \ref{lem:W11}, we have a relation between $W^{(0)}_{2,0}$ and $W^{(0)}_{1,1}$ on the $z$-plane. Pushing it forward to the $x$-plane yields
		\begin{align*}
			W^{(0)}_{2,0}(x_1,x_2)=-\frac{dY(x_2)}{dx_2}W^{(0)}_{1,1}(x_1|Y(x_2)).
		\end{align*}
	The loop insertion operator $\frac{\partial}{\partial V_1(x_3)}$ generates through Leibniz and chain rule
	\begin{align*}
		W^{(0)}_{3,0}(x_1,x_2,x_3)=&-\frac{dW^{(0)}_{2,0}(x_2,x_3)}{dx_2}W^{(0)}_{1,1}(x_1|Y(x_2))-\frac{dY(x_2)}{dx_2}W^{(0)}_{2,1}(x_1,x_3|Y(x_2))\\
		&-\frac{dY(x_2)}{dx_2}\frac{d W^{(0)}_{1,1}(x_1|y)}{dy}\bigg\vert_{y=Y(x_2)}W^{(0)}_{2,0}(x_2,x_3)\\
		=&-\frac{d}{dx_2}\bigg(W^{(0)}_{2,0}(x_2,x_3)W^{(0)}_{1,1}(x_1|Y(x_2))\bigg)-\frac{dY(x_2)}{dx_2}W^{(0)}_{2,1}(x_1,x_3|Y(x_2)),
	\end{align*}
	which is essentially the first relation. The second relation is achieved by using the first one (for $x,y$ interchanged) and 
	\begin{align}\label{W22MM}
		W^{(0)}_{2,0}(x_1,x_2)=\frac{dY(x_1)}{dx_1}\frac{dY(x_2)}{dx_2}W^{(0)}_{0,2}(Y(x_1),Y(x_2))
	\end{align}
	coming from $\omega^{(0)}_{2,0}(z_1,z_2)=\omega^{(0)}_{0,2}(z_1,z_2)$. Again, applying the loop insertion operator $\frac{\partial}{\partial V_1(x_3)}$ yields 
	\begin{align*}
		W^{(0)}_{3,0}(x_1,x_2,x_3)=&\frac{dW^{(0)}_{2,0}(x_1,x_3)}{dx_1}\frac{dY(x_2)}{dx_2}W^{(0)}_{0,2}(Y(x_1),Y(x_2))\\
		+&\frac{dY(x_1)}{dx_1}\frac{dW^{(0)}_{2,0}(x_2,x_3)}{dx_2}W^{(0)}_{0,2}(Y(x_1),Y(x_2))\\
		+&\frac{dY(x_1)}{dx_1}\frac{dY(x_2)}{dx_2}\frac{d W^{(0)}_{0,2}(y,Y(x_2))}{dy}\bigg\vert_{y=Y(x_1)}W^{(0)}_{2,0}(x_1,x_3)\\
		+&\frac{dY(x_1)}{dx_1}\frac{dY(x_2)}{dx_2}\frac{d W^{(0)}_{0,2}(Y(x_1),y)}{dy}\bigg\vert_{y=Y(x_2)}W^{(0)}_{2,0}(x_2,x_3)\\
		+&\frac{dY(x_1)}{dx_1}\frac{dY(x_2)}{dx_2} W^{(0)}_{1,2}(x_3|y,Y(x_2)).
	\end{align*}
	Collecting the terms, applying the chain rule, using \eqref{W22MM}, the first relation \eqref{W122MM} and multiplying by $\frac{dx_1}{dY(x_1)}\frac{dx_2}{dY(x_2)}\frac{dx_3}{dY(x_3)}$ yields the assertion.
	\end{proof}
\end{proposition} 
With Lemma \ref{lem:W11} and Proposition \ref{prop:W32MM}, we observe a general procedure to generate functional relations between $W^{(g)}_{n,m}$. These relations are achieved recursively through an action of a loop insertion operator on a functional relation between $W^{(g)}_{1,0}$ and $W^{(g)}_{0,1}$ . We emphasise that the action of a loop insertion operator was already used in \cite[App. A]{Chekhov:2006vd} to produce some results, however their action of the loop insertion operator was somehow defined on the $z$-plane. Surprisingly, our later Proposition \ref{prop:Wn1} for $W^{(0)}_{n,1}$ is in contradiction to \cite[Theorem A.1]{Chekhov:2006vd} for genus $g=0$. The terms given in \cite[Theorem A.1]{Chekhov:2006vd} consist of a subset of the terms of Proposition \ref{prop:Wn1}.

In the rest of the article, we generalise the situation and treat the loop insertion operator and its action as formal objects. In other words, all results are true for any TR with an existing loop insertion operator and with a spectral curve regular at the ramification points and not coinciding ramifications of $x$ and $y$.

\section{Loop Insertion Operator and Correlators}
The existence of a loop insertion operator for a spectral curve in TR is predicted by deformation theory, see \cite[Theorem 5.2]{Eynard:2007kz}. We will apply the action of the loop insertion operator with eyes on the $x$-$y$ symmetry. \\

Let $\omega^{(g)}_{n,0}$ be the family of meromorphic differentials generated by TR \eqref{eq:trx} from the spectral curve $(\Sigma,x,y,B)$. Changing the role, the spectral curve $(\Sigma,y,x,B)$ defines $\omega^{(g)}_{0,m}$, the family of meromorphic differentials generated by TR \eqref{eq:trx} with $x$ and $y$ interchanged. $x$ and $y$ are coverings of Riemann surfaces with ramification points $\alpha_i$ and $\beta_j$, respectively. We will assume for convenience to have a genus zero spectral curve, which means that $x$ and $y$ have a rational parametrisation.

Recall $\omega^{(0)}_{1,0}(z)=y(z)dx(z)$, $\omega^{(0)}_{0,1}(z)=x(z)dy(z)$ and $\omega^{(0)}_{2,0}(z_1,z_2)=\omega^{(0)}_{0,2}(z_1,z_2)=\frac{dz_1\,dz_2}{(z_1-z_2)^2}$ since we are working with a genus zero spectral curve. We define further 
\begin{align*}
	W^{(g)}_{n,0}(x_1(z_1),...,x_n(z_n)):=&\frac{\omega^{(g)}_{n,0}(z_1,...,z_n)}{dx(z_1)...dx(z_n)}\\
	W^{(g)}_{0,m}(y_1(z_1),...,y_m(z_m)):=&\frac{\omega^{(g)}_{0,m}(z_1,...,z_m)}{dy(z_1)...dy(z_m)}
\end{align*}
with the special case $W^{(0)}_{1,0}(x)=Y(x)$ and $W^{(0)}_{0,1}(y)=X(y)$, where $X(y)$ is by definition the functional inverse of $Y(x)$.

The two coverings $x(z)$ and $y(z)$ are satifying an algebraic equation
\begin{align*}
	E(x(z),y(z))=E(x,Y(x))=E(X(y),y)=0
\end{align*}
depending on the moduli of the spectral curve (see \cite[\S 3.4]{Eynard:2007kz}). Subsequently, if we write (small) $x$ it is understood to be fixed on the spectral curve, which means that (kapital) $Y$ depends on the moduli of the spectral curve. Vice versa, if we write (small) $y$ if is understoot to be fixed on the spectral curve, which means that (kapital) $X$ depends on the moduli of the spectral curve. The veryfies the third property of the upcoming defintion, since all $x_i$ and $y_j$ are fixed on the spectral curve.

Now, we define the main tool for later derivations
\begin{definition}\label{def:op}
	The loop insertion operators $D_x$ and $\tilde{D}_y$ satisfy the Leibniz and chain rule, and the functional equations
	\begin{align}
		&D_xW^{(g)}_{n,0}(x_1,...,x_n)=W^{(g)}_{n+1,0}(x_1,...,x_n,x)\\
		&\tilde{D}_yW^{(g)}_{0,m}(y_1,...,y_m)=W^{(g)}_{0,m+1}(y_1,...,y_m,y)\\
		&D_xx_i=D_xy_j=\tilde{D}_yx_i=\tilde{D}_yy_j=0.\label{Dxx}
	\end{align}
\end{definition}

We give an example for the chain rule 
\begin{example}
	Let $f(x)$ have a non-trivial action of the loop insertion operator, then
	\begin{align*}
		&D_x\big[W^{(g)}_{n,0}(x_1,...,f(x_n))\big]\\
		=&W^{(g)}_{n+1,0}(x_1,...,f(x_n),x)+\frac{d}{d \tilde{x}}W^{(g)}_{n,0}(x_1,...,\tilde{x})\vert_{\tilde{x}=f(x_n)}\cdot D_xf(x_n).
	\end{align*}
	Inserting for instance $f(x)=Y(x)$, we would have
	\begin{align*}
		&D_x\big[W^{(g)}_{n,0}(x_1,...,Y(x_n))\big]\\
		=&W^{(g)}_{n+1,0}(x_1,...,Y(x_n),x)+\frac{d}{d \tilde{x}}W^{(g)}_{n,0}(x_1,...,\tilde{x})\vert_{\tilde{x}=Y(x_n)}\cdot W^{(0)}_{2,0}(x_n,x).
	\end{align*}
\end{example}

To get a well-defined transition from the correlator $W^{(g)}_{n,0}$ to $W^{(g)}_{0,m}$, we need the following assumption
\begin{assumption}\label{Ass}
	The loop insertion operators are commuting
	\begin{align}
		D_x\tilde{D}_y=\tilde{D}_yD_x.
	\end{align}
\end{assumption}
It is not clear if this assumption holds for any spectral curve, this is why it has to be included here. Nevertheless, the assumption is true for the 2-matrix model, since we have two distinct sets of moduli of the spectral curve $(t_i^{(1)})_i$ and $(t_j^{(2)})_j$, where the loop insertion operator $D_x$ is a differential operator wrt to $(t_i^{(1)})_i$ and $\tilde{D}_y$ wrt to $(t_j^{(2)})_j$ (see \eqref{insertion2MM}). Since all $t_i^{(j)}$ are independent parameters, both operators commute.

Due to Assumption \ref{Ass}, the intermediate correlators are well-defined and given by
\begin{definition}\label{def:Wnm}
	Let the correlators $W^{(g)}_{n,m}$ be defined by
	\begin{align*}
		W^{(g)}_{n,m}(x_1,...,x_n|y_1,...,y_m):=&D_{x_2}...D_{x_{n}}\tilde{D}_{y_1}...\tilde{D}_{y_m}W^{(g)}_{1,0}(x_1)\\
		=&D_{x_1}...D_{x_{n}}\tilde{D}_{y_2}...\tilde{D}_{y_{m}}W^{(g)}_{0,1}(y_1).
	\end{align*}
\end{definition}

\begin{example}
	First examples of intermediate correlators are
	\begin{align*}
		W^{(0)}_{1,1}(x|y)=-\frac{dX(y)}{dy}W^{(0)}_{2,0}(x,X(y))
	\end{align*}
\begin{align*}
	W^{(0)}_{2,1}(x_1,x_2|y)=&-\frac{dX(y)}{dy}W^{(0)}_{3,0}(x_1,x_2,X(y))\\\nonumber
	&-\frac{d}{dy}\bigg(\frac{dX(y)}{dy}W^{(0)}_{2,0}(x_2,X(y))W^{(0)}_{2,0}(x_1,X(y))\bigg)
\end{align*}
	which is proved by the properties of the loop insertion operator exactly as in Lemma \ref{lem:W11} and Proposition \ref{prop:W32MM}.
\end{example}

\section{Functional Relation}
\subsection{Trees}
The functional relation has a tree-structure interpretation. To avoid sums of certain restricted sets of subsets, we are going to define the following decorated trees:
\begin{definition}[Trees]\label{def:tree}
	Let $N=\{i_1,...,i_n\}$ and $M=\{j_1,...,j_m\}$.
	Let $\mathcal{G}_{n,m}(N,M)$ be the set of connected trees $T$ with $n$ $\Box$-vertices, $m$ $\circ$-vertices and $\bullet$-vertices, such that 
	\begin{itemize}
		\item the  $\Box$-vertices are labelled by $i_1,...,i_n$
		\item the  $\circ$-vertices are labelled by $j_1,...,j_m$
		\item a $\bullet$-vertex has valence $\geq 2$
		\item edges are only connecting black with white vertices
		\item a  $\Box$-vertex has valence one.
	\end{itemize}
	For a tree $T\in \mathcal{G}_{n,m}(N,M)$, let $r_{j}(T)$ be the valence of the $j^{\text{th}}$ $\circ$-vertex. 
	
	Let $(I,J)$ be a double set associated to a $\bullet$-vertex with $I\subset N$ and $J\subset M$, where $I$ is the set labellings of $\Box$-vertices connected to the $\bullet$-vertex and $J$ the set of labellings of $\circ$-vertices connected to the $\bullet$-vertex. Let $\mathcal{I}(T)$ be the set of all double sets $(I,J)$ for a tree $T\in \mathcal{G}_{n,m}(N,M)$.
\end{definition}
\begin{remark}
	The trees defined by Definition \ref{def:tree} are a generalisation of the set $\mathcal{G}_{n}$ \cite[Def. 3.2]{Borot:2021thu} restricted to Betti number 0, where we additionally allow univalent $\Box$-vertices.  
\end{remark}
We abbreviate $\mathcal{G}_{n,m}=\mathcal{G}_{n,m}(N,M)$, if $N=\{1,...,n\}$ and $M=\{1,...,m\}$.
The set $\mathcal{G}_{n,m}$ is a finite set. It follows from the properties of a connected tree $T\in \mathcal{G}_{n,m}$ that for $m>0$ each $\bullet$-vertex is necessarily connected to at least one  $\circ$-vertex. 
\begin{figure}[htb]
	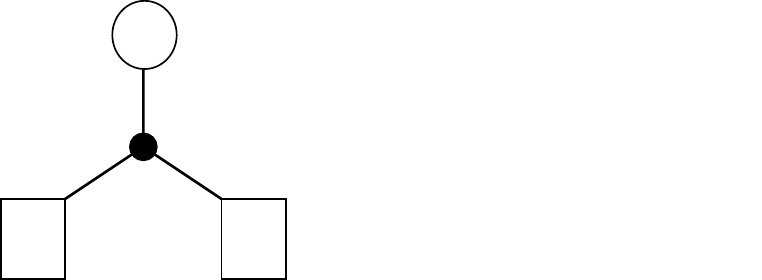
	\caption{The set $\mathcal{G}_{2,1}$ with one  $\circ$-vertex and two  $\Box$-vertices consists of two trees. The $\bullet$-vertex of the first tree $T_1$ is characterised by $\mathcal{I}(T_1)=(I,J)=(\{1,2\},\{1\})$ and the two $\bullet$-vertices of the second tree $T_2$ by $\mathcal{I}(T_2)=\{(I_1,J_1),(I_2,J_2)\}=\{(\{1\},\{1\}),(\{2\},\{1\})\}$}
\end{figure}
\begin{figure}[htb]
	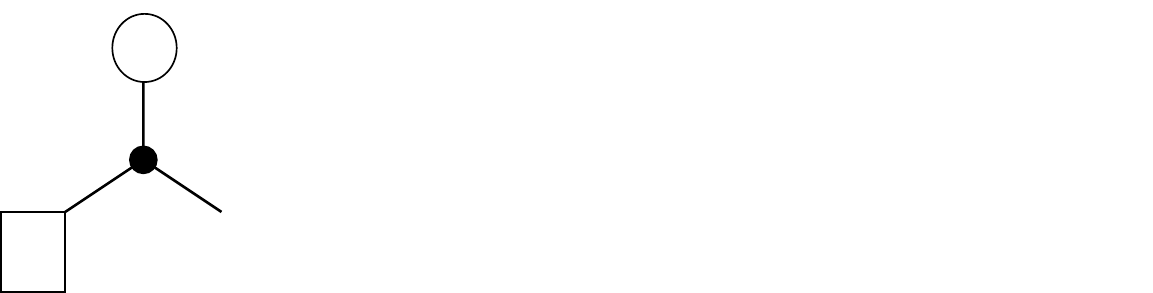
	\caption{The set $\mathcal{G}_{1,2}$ with two  $\circ$-vertices and one  $\Box$-vertex consists of 3 trees. The $\bullet$-vertex of the first tree $T_1$ is characterised by $\mathcal{I}(T_1)=(I,J)=(\{1\},\{1,2\})$, the two $\bullet$-vertices of the second tree $T_2$ by $\mathcal{I}(T_2)=\{(I_1,J_1),(I_2,J_2)\}=\{(\{1\},\{1\}),(\emptyset,\{1,2\})\}$  and the two $\bullet$-vertices of the third tree $T_3$ by $\mathcal{I}(T_3)=\{(I_1,J_1),(I_2,J_2)\}=\{(\emptyset,\{1,2\}),(\{1\},\{2\})\}$} 
\end{figure}

For a tree $T\in \mathcal{G}_{n,m}$, a $\bullet$-vertex is uniquely characterised by the double set $(I,J)$. The set $J$ is a non-empty set (except for $\mathcal{G}_{n,0}$ trees), whereas $I$ could also be empty. The set $\mathcal{I}(T)$ characterises all $\bullet$-vertices and consequently also $T$ uniquely. 

The following example will be of significant interest for later purposes:
\begin{example}\label{ex:Gn1}
	We construct the set $\mathcal{G}_{n,1}\ni T$ with the sets $\mathcal{I}(T)$ of double sets $(I,J)$ corresponding to the $\bullet$-vertices. Any tree $T$ is uniquely characterised by $\mathcal{I}(T)$. Since any $\bullet$-vertex is connected to at least one $\circ$-vertex, it follows that all $\bullet$-vertices are connected to the same $\circ$-vertex labelled by 1. Any $\bullet$-vertex is described by $(I,\{1\})$ with $I\subset \{1,...,n\}$ not empty. If a tree $T\in  \mathcal{G}_{n,1}$ has $k$ $\bullet$-vertices, we have a decomposition $(I_1,\{1\}),(I_2,\{1\}),...,(I_k,\{1\})$, where $I_1\uplus I_2\uplus... \uplus I_k=\{1,...,n\}$ (see Fig \ref{Fig:G31} for $\mathcal{G}_{3,1}$, the left tree has $k=1$ the middle tree $k=2$ and the right $k=3$). Summing over all possibilities, the set $\mathcal{G}_{n,1}$ is isomorphic  to 
	\begin{align*}
		\mathcal{G}_{n,1}\cong \bigg\{\{I_1,...,I_k\}\bigg\vert k\in \{1,...,n\}, I_1\uplus I_2\uplus ...\uplus I_k=\{1,...,n\}\bigg\}\bigg/\sim
	\end{align*}
	where $\{I_1,...,I_k\}\sim \{I'_1,...,I'_k\}$ if $\{I'_1,...,I'_k\}$ is a permutation of $\{I_1,...,I_k\}$ which happens $k!$ times.
	\begin{figure}[htb]
		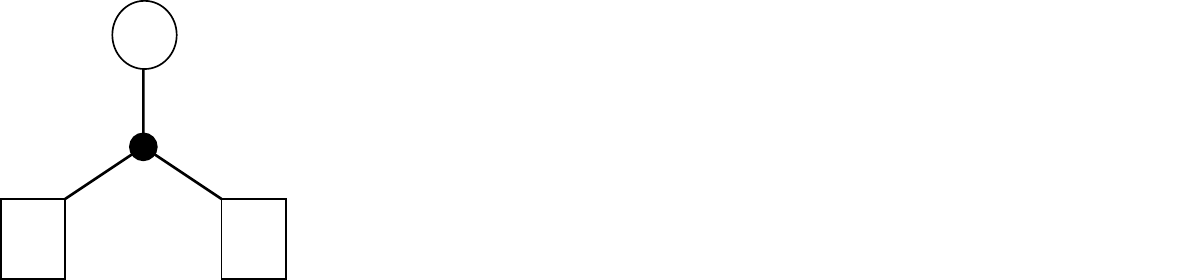
		\caption{The set $\mathcal{G}_{3,1}$ consists of 5 trees. The left tree $T_1$ has the double set $\mathcal{I}(T_1)=(\{1,2,3\},\{1\})$. The middle tree $T_{2,3,4}$ appears three times depending on $i,j$ and $k$, these three cases are $\mathcal{I}(T_2)=\{(\{1,2\},\{1\}),(\{3\},\{1\})\}$ and $\mathcal{I}(T_3)=\{(\{1,3\},\{1\}),(\{2\},\{1\})\}$ and $\mathcal{I}(T_4)=\{(\{2,3\},\{1\}),(\{1\},\{1\})\}$. The right tree $T_5$ is characterised by $\mathcal{I}(T_5)=\{(\{1\},\{1\}),(\{2\},\{1\}),(\{3\},\{1\})\}$.
			\label{Fig:G31}} 
	\end{figure}
\end{example}

The following two lemmata give a recursive construction of the sets $\mathcal{G}_{n+1,m}$ and $\mathcal{G}_{n,m+1}$ depending on $\mathcal{G}_{n',m'}$ with $n'\leq n$ and $m'\leq m$. 
\begin{lemma}[$\mathcal{G}_{n+1,m}$]\label{lem:Gn1m}
	The set $\mathcal{G}_{n+1,m}$ is constructed by connecting to any $T\in \mathcal{G}_{n,m}$ a $\Box$-vertex labelled by $n+1$ either
	\begin{itemize}
		\item to a  $\circ$-vertex $k\in \{1,...,m\}$ via a $\bullet$-vertex of valence two
		\item or to a $\bullet$-vertex characterised by $(I,J)$
	\end{itemize}
	of $T$ in all possible ways.
	
	The $\bullet$-vertex, which is now connected to the  $\Box$-vertex labelled by $n+1$,  is characterised by the double set 
	\begin{itemize}
		\item $(\{n+1\},\{k\})$ if it is two-valent
		\item $(\{I,n+1\},J)$ if it has valence $>2$.
	\end{itemize}
	\begin{proof}
		The connected $\Box$-vertex labelled by $n+1$ is univalent and necessarily connected to a $\bullet$-vertex by Definition \ref{def:tree}. If the $\bullet$-vertex is two-valent, it is further connected to a $\circ$-vertex. If the $\bullet$-vertex has valence $>2$, it is further connected to at least one  $\circ$-vertex. 
		
		\begin{figure}[htb]
			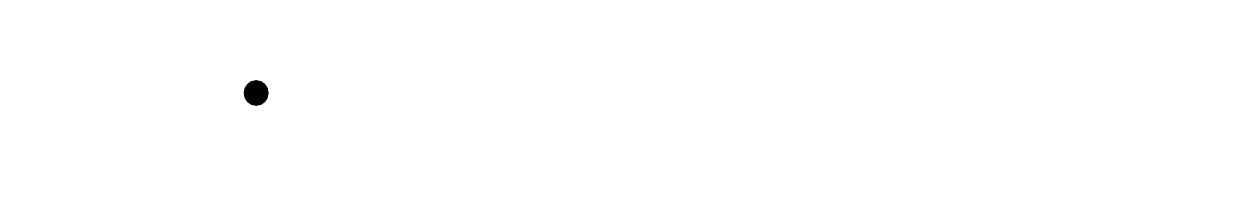
			\caption{Any tree of $\mathcal{G}_{n+1,m}$ is constructed by a tree $T\in \mathcal{G}_{n,m}$ by connecting the  $\Box$-vertex labelled by $n+1$ to a  $\circ$-vertex $k$ of $T$ via a two-valent $\bullet$-vertex (left) or to a $\bullet$-vertex of $T$ \label{Fig:Gn1m}} 
		\end{figure}
		Both operations together generate any tree in $\mathcal{G}_{n+1,m}$ exactly once, which can be seen by deletion. 
		
		Surjectivity: Delete the edge of the  $\Box$-vertex labelled by $n+1$ (see Fig. \ref{Fig:Gn1m}) (and delete also its connected $\bullet$-vertex together with its edge if the $\bullet$-vertex was two-valent), a tree of the set $\mathcal{G}_{n,m}$ remains. 
		
		Injectivity: Both operations together generate any tree in $\mathcal{G}_{n+1,m}$ exactly one time, since the previously described deletion generates a unique tree in $\mathcal{G}_{n,m}$.
	\end{proof}
\end{lemma}

The construction of $\mathcal{G}_{n,m+1}$ is a bit more tricky, because the additional $\circ$-vertex can have any valence $\geq 1$, whereas the additional $\Box$-vertex in Lemma \ref{lem:Gn1m} had just the valence one due to Definition \ref{def:tree}. The crucial step in constructing $\mathcal{G}_{n,m+1}$ from $\mathcal{G}_{n,m}$ is to recognise the partial substructure of $\mathcal{G}_{k,1}$ of Example \ref{ex:Gn1}.

\begin{lemma}[$\mathcal{G}_{n+1,m}$]\label{lem:Gnm1a}
	The set $\mathcal{G}_{n,m+1}$ is constructed by connecting to any $T\in \mathcal{G}_{n,m}$ a $\circ$-vertex labelled by $m+1$ either
	\begin{itemize}
		\item to a  $\circ$-vertex $j\in \{1,...,m\}$ via a $\bullet$-vertex of valence two
		\item or to the connected parts of a $\bullet$-vertex of valence $k$ through the structure given by $\mathcal{G}_{k,1}$
	\end{itemize}
	in all possible ways.
	
	The $\circ$-vertex labelled by $m+1$ is connected to
	\begin{itemize}
		\item the $\bullet$-vertex of valence two characterised by $(
		\emptyset,\{j,m+1\})$ 
		\item a set of $\bullet$-vertices characterised by an element of $\mathcal{G}_{k,1}$.
	\end{itemize}
	\begin{proof}
		We discuss both cases described in the lemma separately. 
	
		The first case is the analogous situation to Lemma \ref{lem:Gn1m} generated by the full subset of $\mathcal{G}_{n,m+1}$, where the $\circ$-vertex labelled by $m+1$ is connected to a $\bullet$-vertex of valence two.
		
		The second case is more involved. Connecting a $\circ$-vertex through $\bullet$-vertices to $k$ further white vertices (either $\circ$- or $\Box$-vertices) is characterised by the set $\mathcal{G}_{k,1}$. To see this, let a $\bullet$-vertex of valence $k$ be connected to the white vertices of $T_1,...,T_k$ trees (see the upper tree in Fig. \ref{Fig:Gnm1a}). Permuting the subtrees $T_i$ will not change the entire tree. Each of those subtrees $T_i$ can therefore be interpreted as a $\Box$-vertex such that the entire tree has an interpretation as the unique tree in $\mathcal{G}_{k,0}$, since it has as a $\Box$-vertex a valence of one and is not further connected. Now, adding a $\circ$-vertex to the unique vertex in $\mathcal{G}_{k,0}$ generates the set $\mathcal{G}_{k,1}$ described in Example \ref{ex:Gn1} already.  
		\begin{figure}[htb]
		\scalebox{.85}{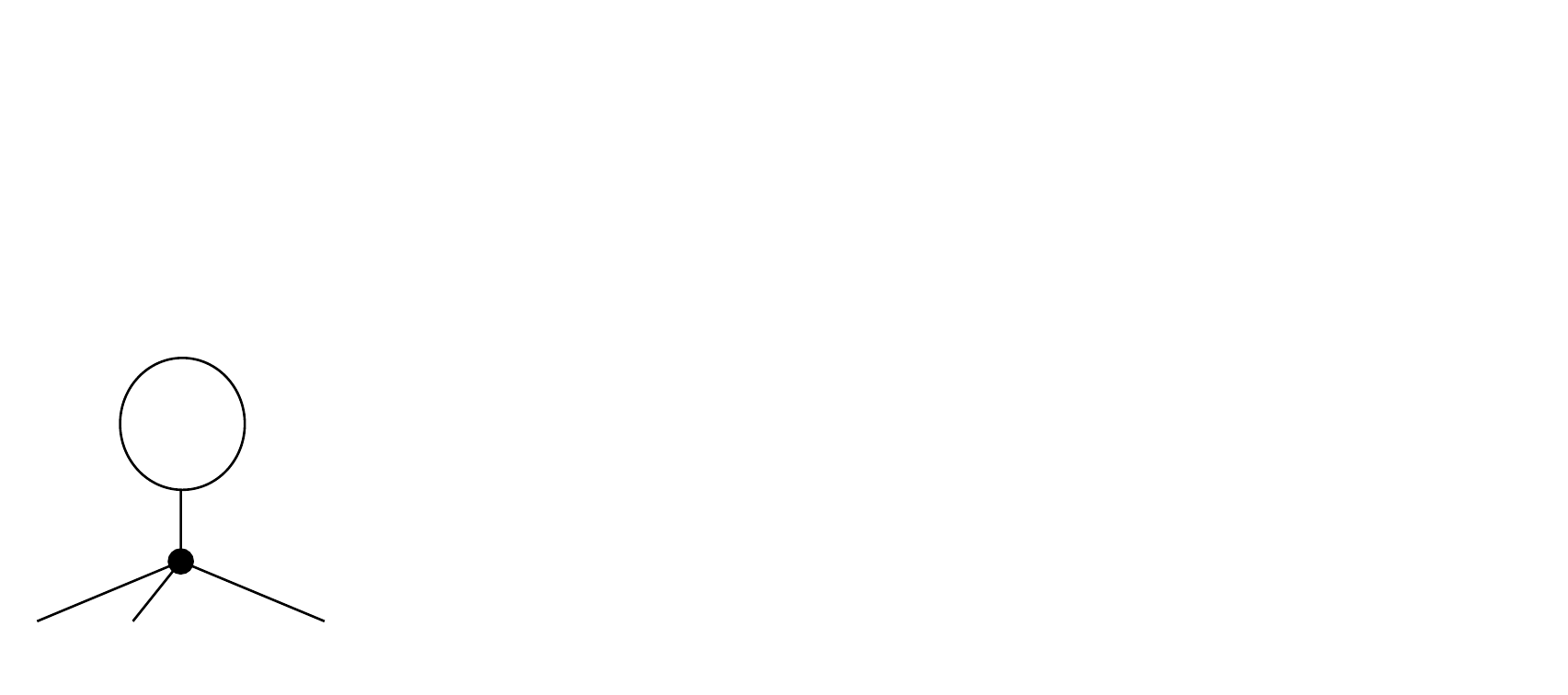}
		\caption{For a tree $T\in \mathcal{G}_{n,m}$, any $\bullet$-vertex of valence $k$ can be interpreted as the unique element of $\mathcal{G}_{k,0}$ where the subtrees $T_i$ are just $\Box$-vertices (upper tree). Adding the $\circ$-vertex labelled by $m+1$ can be interpreted as generating the full set $\mathcal{G}_{k,1}$, where the subtrees $T_i$ remains $\Box$-vertices (below trees).
			\label{Fig:Gnm1a}} 
		\end{figure}
		In particular, all the 
	 	trees below in Fig. \ref{Fig:Gnm1a} are generated, where the $\bullet$-vertices are characterised by the elements in $\mathcal{G}_{k,1}$.
	
		To see that both operations together give a bijection to $\mathcal{G}_{n+1,m}$, we will look at the inverse operation, deletion. 
		
		Surjectivity: Take an element $T\in \mathcal{G}_{n,m+1}$, which is connected through $k'$ $\bullet$-vertices to $k$ white vertices (either $\circ$- or $\Box$-vertices). All $k'$ $\bullet$ vertices are characterised by some double sets $(I_i,\{m+1\})_{i=1,...,k'}$ (if we collect the subtrees to a $\Box$-vertex). The set of all these double sets is $\mathcal{I}(\tilde{T})\in \mathcal{G}_{k,1}$. Next, delete all edges connected to the $\circ$-vertex labelled by $m+1$  (and delete also its connected $\bullet$-vertex together with its edge if the $\bullet$-vertex is now univalent) and merge all $\bullet$-vertices which were connected to the deleted edges. This generates an element in $\mathcal{G}_{n,m}$. 
		
		Injectivity: It might be possible that the previously described deletion operation generates also the upper drawn tree in Fig. \ref{Fig:Gnm1a} with $T_1',...,T_k'$ as a permutation of $T_1,...,T_k$. However, these two trees are the same. 
	\end{proof}
\end{lemma}

The two Lemmata \ref{lem:Gn1m} and $\ref{lem:Gnm1a}$ are the combinatorial constructions in the later proof by induction of the functional relation of $W^{(0)}_{n,m}$. Assume we know the functional relation of $W^{(0)}_{n,m}$, then we will need Lemma \ref{lem:Gn1m} to prove the functional relation of $W^{(0)}_{n+1,m}$, and  we will need Lemma \ref{lem:Gnm1a} to prove the functional relation of $W^{(0)}_{n,m+1}$.

\subsection{Main Theorem}
To prove the main result, it is useful to associate a weight $\phi$ to a tree $T\in \mathcal{G}_{n,m}$ by  
\begin{definition}[Weight]
	Let $T\in \mathcal{G}_{n,m}$. We associate to the $i^{\text{th}}$ $\Box$-vertex the variable $x_i$ and to the $j^{\text{th}}$  $\circ$-vertex the variable $y_j$. Then, we define the weight of a tree by
	\begin{align*}
		\phi(T):=\prod_{j=1}^m\bigg(-\frac{d}{d y_j}\bigg)^{r_j(T)-1}\bigg(\prod_{k=1}^m\bigg(-\frac{dX(y_k)}{dy_k}\bigg)\prod_{(I,J)\in \mathcal{I}(T)}W^{(0)}_{|I|+|J|,0}\big(x_I,X(y_J)\big)\bigg),
	\end{align*}
	where $r_j(T)$ is the valence of the $j^{\text{th}}$ $\circ$-vertex and $\mathcal{I}(T)$ is the set of double sets $(I,J)$ characterising a $\bullet$-vertex, which is connected to the $\Box$-vertices labelled by the elements of $I\subset \{1,...,n\}$ and to the $\circ$-vertices labelled by the elements of $J\subset \{1,...,m\}$.
\end{definition}
The weight $\phi$ associates to any tree $T\in \mathcal{G}_{n,m}$ a function depending on the variables $x_1,...,x_n,y_1,...,y_m$, which is built by the correlators $W^{(0)}_{k,0}$. To prove the general case, we need first of all the following special case
\begin{proposition}\label{prop:Wn1}
	Consider Assumption \ref{Ass}.
	The functional of $W^{(0)}_{n,1}$ reads
	\begin{align}\label{Wn1}
		W^{(0)}_{n,1}(x_1,...,x_n|y)=&\sum_{T\in \mathcal{G}_{n,1}}\frac{\phi(T)}{\mathrm{Aut}(T)}\\\nonumber
		=&-\sum_{T\in \mathcal{G}_{n,1}}\frac{1}{\mathrm{Aut}(T)}\bigg(-\frac{d}{d y}\bigg)^{r(T)-1}\bigg(\frac{dX(y)}{dy}\prod_{(I,\{y\})\in \mathcal{I}(T)}W^{(0)}_{|I|+1,0}\big(x_I,X(y)\big)\bigg)\\\nonumber
		=&\sum_{k=1}^n\sum_{I_1\uplus I_2\uplus ...\uplus I_k=I}\frac{(-1)^k}{k!}\bigg(\frac{d}{d y}\bigg)^{k-1}\bigg(\frac{dX(y)}{dy}\prod_{i=1}^kW^{(0)}_{|I_i|+1,0}\big(x_{I_i},X(y)\big)\bigg),
	\end{align}
	where $r(T)$ is the valence of the $\circ$-vertex of $T$, and we use the abbreviation $x_I=\{x_{i_1},...,x_{i_{|I|}}\}$ for $I=\{i_1,...,i_{|I|}\}$.
\begin{proof}
	The first representation is equivalent to the second by the definition of the weight $\phi$. The second representation is equivalent to the third by the isomorphism described in Example \ref{ex:Gn1}.
	
	The functional relation is proved by induction in $n$. The action of the loop insertion operator will be interpreted graphically such that we can use Lemma \ref{lem:Gn1m}.
	
	The initial case with $n=1$ is true due to Lemma \ref{lem:W11}
	\begin{align*}
		W^{(0)}_{1,1}(x|y)=-\frac{dX(y)}{dy}W^{(0)}_{2,0}(x,X(y)).
	\end{align*}
	
	Assume \eqref{Wn1} is true for $n$, then the loop insertion operator $D_{x_{n+1}}$ acts on the left as
	\begin{align*}
		D_{x_{n+1}}W^{(0)}_{n,1}(x_I|y)=W^{(0)}_{n+1,1}(x_I,x_{n+1}|y).
	\end{align*}
	On the right, the loop insertion operator commutes with all derivatives wrt to $y$ due to the property \eqref{Dxx}. It remains to compute 
	\begin{align}\label{Dx1}
		D_{x_{n+1}}\bigg(\frac{dX(y)}{dy}\prod_{(I,\{y\})\in \mathcal{I}(T)}W^{(0)}_{|I|+1,0}\big(x_I,X(y)\big)\bigg).
	\end{align}
	First, we look at the action of the loop insertion operator $D_{x_{n+1}}$ by chain rule on $X(y)$ in \eqref{Dx1}, which gives
	\begin{align*}
		&\bigg(\frac{dW^{(0)}_{1,1}(x_{n+1}|y)}{dy}\prod_{(I,\{y\})\in \mathcal{I}(T)}W^{(0)}_{|I|+1,0}\big(x_I,X(y)\big)\bigg)\\
		+& \bigg(\frac{dX(y)}{dy}\!\!\!\!\!\!\sum_{(I,\{y\})\in \mathcal{I}(T)}\!\!\!\!\!\!\frac{d W^{(0)}_{|I|+1,0}\big(x_I,\tilde{x}\big)}{d\tilde{x}}\bigg\vert_{\tilde{x}=X(y)}\!\!\!\!\!\! W^{(0)}_{1,1}(x_{n+1}|y)\!\!\!\!\!\!\!\!\!\!\!\! \prod_{(I',\{y\})\in \mathcal{I}(T)\setminus (I,\{y\}) }\!\!\!\!\!\!\!\!\!\!\!\! W^{(0)}_{|I|+1,0}\big(x_I,X(y)\big)\bigg)\\
		=&-\bigg(\frac{d\big(W^{(0)}_{2,0}(x_{n+1},X(y)) \cdot \frac{dX(y)}{dy}\big)}{dy}\prod_{(I,\{y\})\in \mathcal{I}(T)}W^{(0)}_{|I|+1,0}\big(x_I,X(y)\big)\bigg)\\
		-& \bigg(\frac{dX(y)}{dy}\sum_{(I,\{y\})\in \mathcal{I}(T)}\frac{d W^{(0)}_{|I|+1,0}\big(x_I,\tilde{x}\big)}{d\tilde{x}}\bigg\vert_{\tilde{x}=X(y)}\frac{dX(y)}{dy} W^{(0)}_{2,0}(x_{n+1},X(y)) \\
		 &\qquad\qquad \qquad\qquad \qquad \qquad \times \prod_{(I',\{y\})\in \mathcal{I}(T)\setminus (I,\{y\}) }W^{(0)}_{|I|+1,0}\big(x_I,X(y)\big)\bigg)\\
		 =&-\frac{d}{dy}\bigg(\frac{dX(y)}{dy} W^{(0)}_{2,0}(x_{n+1},X(y))\prod_{(I,\{y\})\in \mathcal{I}(T)}W^{(0)}_{|I|+1,0}\big(x_I,X(y)\big)\bigg),
	\end{align*}
	where we have used Lemma \ref{lem:W11}, the chain rule $ \frac{dX(y)}{dy}\frac{d f(\tilde{x})}{d\tilde{x}}\bigg\vert_{\tilde{x}=X(y)}=\frac{f(X(y))}{dy}$ and the Leibniz rule in the last line. 
	
	This computation has got a graphical interpretation due to the weight $\phi$ (see Fig. \ref{Fig:DxcircWn1}).
	\begin{figure}[htb]
		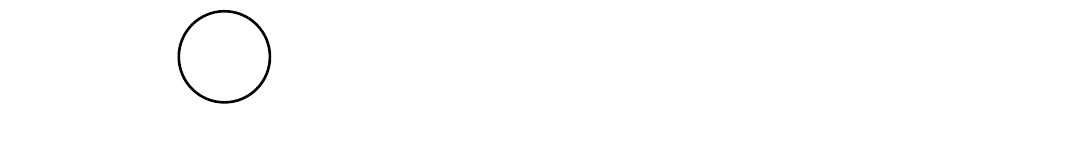
		\caption{The action of the loop insertion operator $D_{x_{n+1}}$ on the $\circ$-vertex. A further edge connected to a $\bullet$-vertex is generated which is connected to a $\Box$-vertex labelled by $n+1$.\label{Fig:DxcircWn1}} 
	\end{figure}
	The additional edge connected to the $\circ$-vertex brings an additinal derivative $-\frac{d}{dy}$, the $\bullet$-vertex of valence two generates the factor $W^{(0)}_{2,0}(x_{n+1},X(y))$ and the rest of the tree remains.
	
	Furthermore, the loop insertion operator acts on the $W^{(0)}_{n',0}$ in \eqref{Dx1}, which is just
	\begin{align*}
		\bigg(\frac{dX(y)}{dy}\sum_{(I,\{y\})\in \mathcal{I}(T)} W^{(0)}_{|I|+1,0}\big(x_I,x_{n+1},X(y)\big) \prod_{(I',\{y\})\in \mathcal{I}(T)\setminus (I,\{y\}) }W^{(0)}_{|I|+1,0}\big(x_I,X(y)\big)\bigg).
	\end{align*}
	\begin{figure}[htb]
		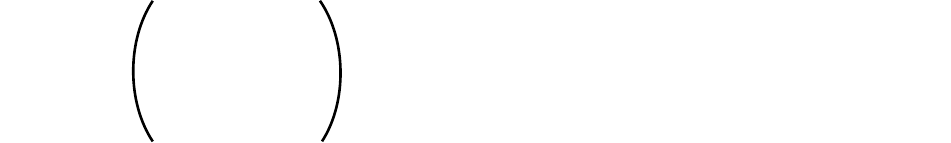
		\caption{The action of $D_{x_{n+1}}$ on a $\bullet$-vertex generates an additional edges connected to the $\Box$-vertex labelled by $n+1$.   
			\label{Fig:Dxbullet}} 
	\end{figure}
	This action is interpreted graphically as an action on a $\bullet$-vertex adding an edge connected to the $\Box$-vertex labelled by $n+1$ (see Fig. \ref{Fig:Dxbullet}).
	
	Both actions together generate graphically exactly the construction described in Lemma \ref{lem:Gn1m}, and therefore the set $\mathcal{G}_{n+1,m}$. This finishes the proof.
\end{proof}
\end{proposition}
The special term related to $k=1$ in \eqref{Wn1} can be pulled out. Pulling all variables back to the $z$-plane by $y=y(z)$ and $x_i=x(z_i)$, we get
\begin{align}\label{2MMVergleich}
	&W^{(0)}_{n,1}(x(z_1),...,x(z_i)|y(z))+\frac{dx(z)}{dy(z)}W^{(0)}_{n+1,0}(x(z_1),...,x(z_n),x(z))\\\nonumber
	=&\sum_{k=2}^n\sum_{I_1\uplus I_2\uplus ...\uplus I_k=I}\frac{(-1)^k}{k!}\bigg(\frac{d}{d y(z)}\bigg)^{k-1}\bigg(\frac{dx(z)}{dy(z)}\prod_{i=1}^kW^{(0)}_{|I_i|+1,0}\big(x(z_1),...,x(z_n),x(z)\big)\bigg).
\end{align}

\begin{remark}
	In case of the 2-matrix model, a formula for any $W^{(g)}_{n,1}(x_I|y)$ was provided in \cite[Theorem A.1]{Chekhov:2006vd} which is for $g=0$ in contradiction to Proposition \ref{Wn1}. The proof in \cite{Chekhov:2006vd} uses the loop insertion operator as well, however on the $z$-plane as an action on the residue formula of TR. The loop insertion operator might be not well-defined on the $z$-plane, which causes the difference. We find that Theorem A.1 of \cite{Chekhov:2006vd} for $g=0$ is equal to a subset of the terms of Proposition \ref{prop:Wn1}. More precisely, if we would restrict the valence of the $\circ$-vertices to $\leq 2$ (or equivalently if we just take $k= 2$ in \eqref{2MMVergleich}), we would have the same structure as in \cite[Theorem A.1]{Chekhov:2006vd} for $g=0$.
\end{remark}

The more general case, which needs Proposition \ref{prop:Wn1} as an essential ingredient, is formulated as the main theorem:
\begin{theorem}[Main Theorem]\label{thm:main}
	Consider Assumption \ref{Ass}.
	The functional relation of $W^{(0)}_{n,m}$ reads
	\begin{align}\label{eq:mainthm}
		&W^{(0)}_{n,m}(x_1,...,x_n|y_1,...,y_m)=\sum_{T\in \mathcal{G}_{n,m}}\frac{\phi(T)}{\mathrm{Aut}(T)}\\\nonumber
		=&\sum_{T\in \mathcal{G}_{n,m}}\frac{1}{\mathrm{Aut}(T)}\prod_{j=1}^m\bigg(-\frac{d}{d y_j}\bigg)^{r_j(T)-1}\bigg(\prod_{k=1}^m\bigg(-\frac{dX(y_k)}{dy_k}\bigg)\prod_{(I,J)\in \mathcal{I}(T)}W^{(0)}_{|I|+|J|,0}\big(x_I,X(y_J)\big)\bigg),
	\end{align}
where $r_j(T)$ is the valence of the $j^{\text{th}}$ $\circ$-vertex of $T$, and we use the abbreviation $x_I:=\{x_{i_1},...,x_{i_{|I|}}\}$ and $X(y_J)=\{X(y_{j_1}),...,X(y_{j_{|J|}})\}$ for some $I=\{i_1,...,i_{|I|}\}$ and $J=\{j_1,...,j_{|J|}\}$.
\begin{proof}
	We prove by induction in $n+m\mapsto n+m+1$, where we have to distinguish between $n\mapsto n+1$ and $m\mapsto m+1$. The initial cases are true for $(n,m)=(2,0),(1,1),(0,2)$.
	\begin{itemize}
		\item $W^{(0)}_{n+1,m}$: The step from $W^{(0)}_{n,m}$ to $W^{(0)}_{n+1,m}$ works exactly in the same way as from $W^{(0)}_{n,1}$ to $W^{(0)}_{n+1,1}$ in Proposition \ref{prop:Wn1}. The action of the loop insertion operator $D_{x_{n+1}}$ has to be computed, where we just have to distinguish between the different $y_j$, which means graphically between the different $\circ$-vertices labelled by $j$ (see Fig. \ref{Fig:Wn1mj}).
		\begin{figure}[htb]
			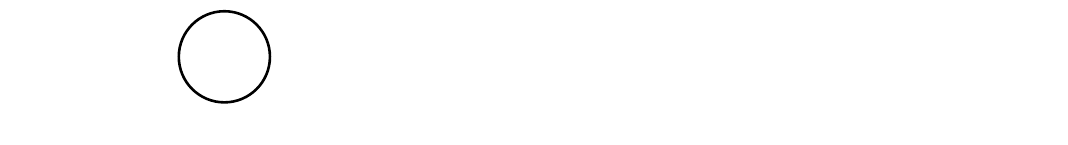
			\caption{Action of the loop insertion operator on $D_{x_{n+1}}$ on a $\circ$-vertex labelled by $j$.\label{Fig:Wn1mj}} 
		\end{figure}
		The action of $D_{x_{n+1}}$ on the $\bullet$-vertex is descibed in the proof of Proposition \ref{prop:Wn1}. 
		
		Both operation together give the construction of $\mathcal{G}_{n+1,m}$ described in Lemma \ref{lem:Gn1m}.
		
		\item $W^{(0)}_{n,m+1}$: The action of the loop insertion operator $\tilde{D}_y$ will be computed in more detail. Let the loop insertion operator $\tilde{D}_{y_{m+1}}$ act on \eqref{eq:mainthm}, the lhs gives by definition $W^{(0)}_{n,m+1}(x_1,..,x_n|y_1,...,y_m,y_{m+1})$. For the rhs, $\tilde{D}_{y_{m+1}}$ commutes with all $y_j$ derivatives by \eqref{Dxx}, so we have to compute
		\begin{align}\label{proof1}
			\tilde{D}_{y_{m+1}}\bigg(\prod_{k=1}^m\bigg(-\frac{dX(y_k)}{dy_k}\bigg)\prod_{(I,J)\in \mathcal{I}(T)}W^{(0)}_{|I|+|J|,0}\big(x_I,X(y_J)\big)\bigg).
		\end{align}
		First, we look only at the action on some $X(y_j)$ in \eqref{proof1} which is
		\begin{align*}
			&\bigg(-\frac{dW^{(0)}_{0,2}(y_j,y_{m+1})}{dy_j}\bigg)\prod_{k\neq j}\bigg(-\frac{dX(y_k)}{dy_k}\bigg) \prod_{(I,J)\in \mathcal{I}(T)}W^{(0)}_{|I|+|J|,0}\big(x_I,X(y_J)\big)\\
			+& \bigg(\prod_{k=1}^m\bigg(-\frac{dX(y_k)}{dy_k}\bigg)\sum_{\substack{(I,J)\in \mathcal{I}(T)\\ j\in J}} \frac{W^{(0)}_{|I|+|J|,0}\big(x_I,X(y_{J\setminus j}),\tilde{x}\big)}{d\tilde{x}}\bigg\vert_{\tilde{x}=X(y_j)}\cdot W^{(0)}_{0,2}(y_j,y_{m+1})
			\\
			&\qquad\qquad\qquad \times\prod_{(I',J')\in \mathcal{I}(T)\setminus (I,J)}W^{(0)}_{|I'|+|J'|,0}\big(x_{I'},X(y_{J'})\big)\bigg).
		\end{align*}
		Insert $W^{(0)}_{0,2}(y_j,y_{m+1})=\frac{dX(y_i)}{dy_j}\frac{dX(y_{m+1})}{dy_{m+1}}W^{(0)}_{2,0}(X(y_j),X(y_{m+1}))$, use the chain rule $ \frac{dX(y_j)}{dy_j}\frac{d f(\tilde{x})}{d\tilde{x}}\bigg\vert_{\tilde{x}=X(y_j)}=\frac{f(X(y_j))}{dy_j}$ and the Leibniz rule yields
		\begin{align*}
			=\bigg(-\frac{d}{dy_j}\bigg)\bigg(\prod_{k=1}^{m+1}\bigg(-\frac{dX(y_k)}{dy_k}\bigg)W^{(0)}_{2,0}(X(y_j),X(y_{m+1}))\prod_{(I,J)\in \mathcal{I}(T)}W^{(0)}_{|I|+|J|,0}\big(x_I,X(y_J)|\big)\bigg).
		\end{align*}
		This action has a graphical interpretation due to the weight $\phi$ shown in Fig. \ref{Fig:Wnm1j}. 
		\begin{figure}[htb]
			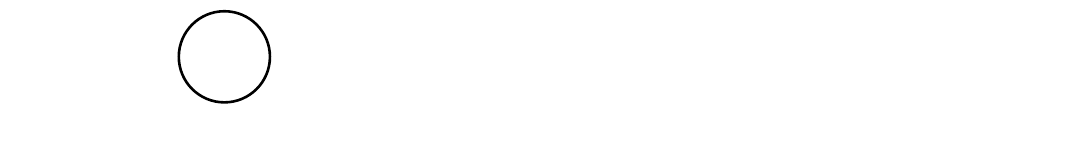
			\caption{Action of the loop insertion operator on $\tilde{D}_{y_{m+1}}$ on a $\circ$-vertex labelled by $j$.\label{Fig:Wnm1j}} 
		\end{figure}
		The additional derivative wrt to $y_j$ gives an additional edge at the $\circ$-vertex labelled by $j$. The $\bullet$-vertex of valence two gives the factor $W^{(0)}_{2,0}(X(y_j),X(y_{m+1}))$ since it is connected to two $\circ$-vertices, and the rest remains untouched.
		
		\begin{figure}[htb]
			\scalebox{.85}{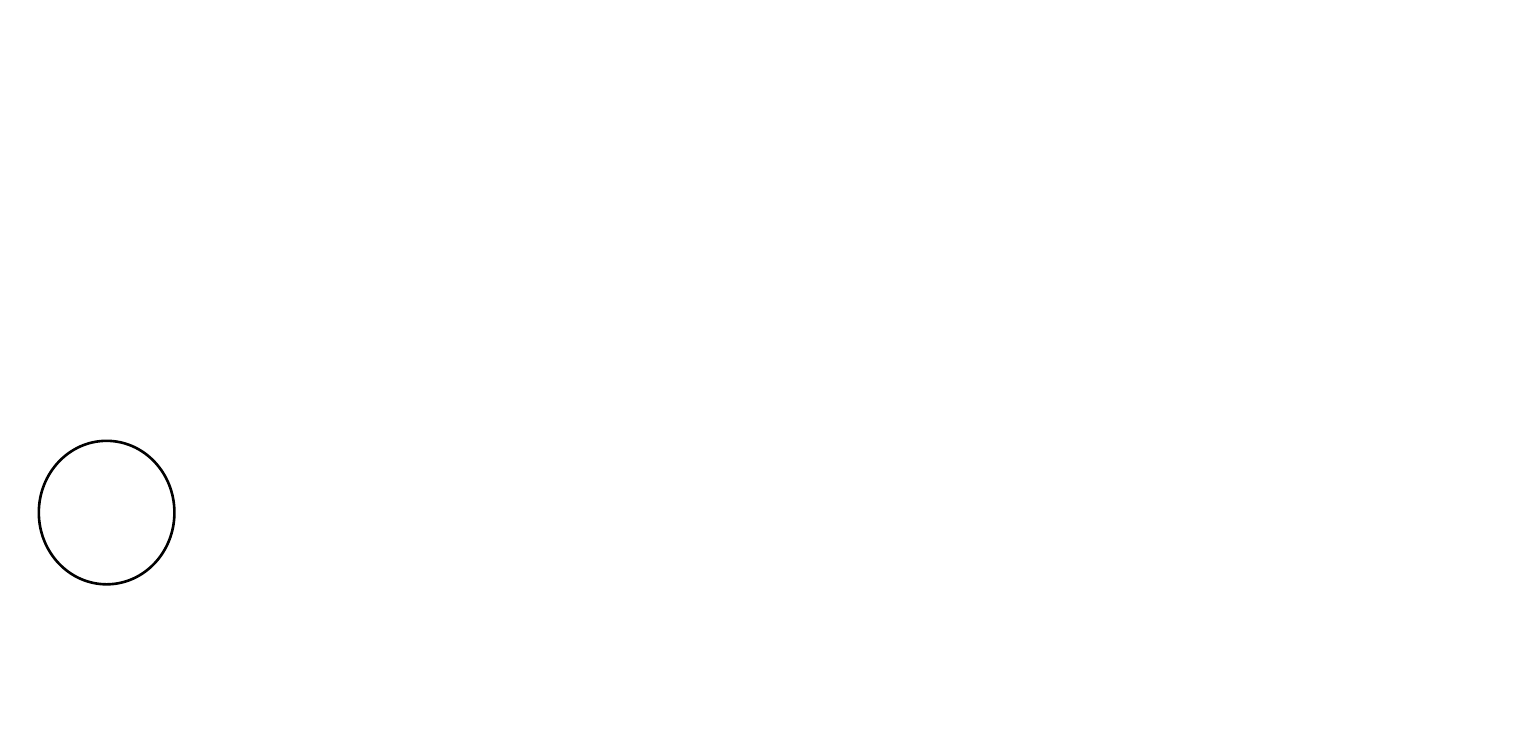}
			\caption{Action of the loop insertion operator on $\tilde{D}_{y_{m+1}}$ on a $\bullet$-vertex of valence $k$ generates the set $\mathcal{G}_{k,1}$.\label{Fig:Gnm1Proof}} 
		\end{figure}
		Second, we look at the action of $\tilde{D}_{y_{m+1}}$ in \eqref{proof1} on some $W^{(0)}_{n',0}$ which is
		\begin{align*}
			\prod_{k=1}^m\bigg(-\frac{dX(y_k)}{dy_k}\bigg)W^{(0)}_{|I|+|J|,1}\big(x_{I},X(y_{J})|y_{m+1}\big) \prod_{(I',J')\in \mathcal{I}(T) \setminus (I,J)}W^{(0)}_{|I'|+|J'|,0}\big(x_{I'},X(y_{J'})\big).
		\end{align*}
	
		A function $W^{(0)}_{|I|+|J|,1}\big(x_{I},X(y_{J})|y_{m+1}\big)$ is generated which was proved in Proposition \ref{prop:Wn1} to be graphically described by the set $\mathcal{G}_{|I|+|J|,1}$ (see Fig. \ref{Fig:Gnm1Proof}). 
		
		All variables $x_I$ and $X(y_J)$ are treated equivalently as a kind of $\Box$-vertices. 
		
		Lemma \ref{lem:Gnm1a} proves that these two actions of $\tilde{D}_{y_{m+1}}$ together construct the full set $\mathcal{G}_{n,m+1}$ uniquely. 
	\end{itemize}
\end{proof}
\end{theorem} 
 We emphasise that the formula of the theorem can be pulled back to the $z$-plane writing $x_i=x(z_i)$ and $y_j=y(w_j)$, where $x(z)$ and $y(w)$ are the two coverings corresponding to the spectral curve $(\Sigma,x,y,B)$. The variables $z_i$ and $w_j$ live on the $z$-plane and are independent of each other. Theorem \ref{thm:main} becomes
\begin{align}\label{Wzp}
	&W^{(0)}_{n,m}(x(z_1),...,x(z_n)|y(w_1),...,y(w_m))\\\nonumber
	=&\sum_{T\in \mathcal{G}_{n,m}}\frac{1}{\mathrm{Aut}(T)}\prod_{j=1}^m\bigg(-\frac{d}{d y(w_j)}\bigg)^{r_j(T)-1}\bigg(\prod_{k=1}^m\bigg(-\frac{dx(w_k)}{dy(w_k)}\bigg)\prod_{(I,J)\in \mathcal{I}(T)}W^{(0)}_{|I|+|J|,0}\big(x(z_I),x(w_J)\big)\bigg).
\end{align}

The special case of $n=0$ was already stated in Theorem \ref{thm:first}.
We know from the theory of TR that $W^{(0)}_{0,m}(y(w_1))$ with $m>2$ has just poles at the ramification point $\beta_i$ of $y(w)$. However, the rhs of \eqref{Wzp} is built by $W^{(0)}_{n',0}(x(w_{i_1}),...,x(w_{i_{n'}}))$ which has poles at the ramification points $\alpha_j$ of $x(w)$ and at the antidiagonal. Consequently, performing the sum over all trees $\mathcal{G}_{0,m}$ cancels those poles, and the derivatives wrt to $y(w_j)$ generate the correct poles at the ramification points of $y(w)$.

\subsection{Examples}
To get a feeling what Theorem \ref{thm:main} or \ref{thm:first} are actually telling, we look at examples
\subsubsection*{Airy Curve}
Take the Airy spectral curve as a prime example defined by $(\mathbb{P}^1,x(z)=z^2,y(z)=z,\frac{dz_1\, dz_2}{(z_1-z_2)^2})$. Following the algorithm of TR \eqref{eq:trx}, we have
\begin{align*}
	\omega^{(0)}_{1,0}(z)=&2z^2 dz\\
	\omega^{(0)}_{2,0}(z_1,z_2)=&\frac{dz_1\, dz_2}{(z_1-z_2)^2}\\
	\omega^{(0)}_{3,0}(z_1,z_2,z_3)=&-\frac{dz_1dz_2dz_3}{2z_1^2z_2^2z_3^2}\\
	\omega^{(0)}_{4,0}(z_1,z_2,z_3,z_4)=&\frac{3dz_1dz_2dz_3dz_4}{4z_1^2z_2^2z_3^2z_4^2}\sum_{i=1}^4\frac{1}{z_i^2}
\end{align*} 
as well as 
\begin{align*}
	\omega^{(0)}_{0,1}(z)=&z^2 dz\\
	\omega^{(0)}_{0,2}(z_1,z_2)=&\frac{dz_1\, dz_2}{(z_1-z_2)^2}
\end{align*}
and $\omega^{(0)}_{0,m}=0$ for $m>2$, since $y(z)=z$ has no ramification point.

For $m=3$, $\mathcal{G}_{0,3}$ consists of $4$ trees and Theorem \ref{thm:first} becomes
\begin{align*}
	0=&W^{(0)}_{0,3}(y(z_1),y(z_2),y(z_3))\\
	=&\frac{-\prod_{i=1}^3(-2z_i)}{16z_1^3z_2^3z_3^3}-\frac{d}{dz_1}\bigg(\frac{-8 z_1z_2z_3}{4z_1z_2(z_1-z_2)^24z_1z_3(z_1-z_3)^2}\bigg)\\-
	&\frac{d}{dz_2}\bigg(\frac{-8 z_1z_2z_3}{4z_1z_2(z_1-z_2)^24z_2z_3(z_2-z_3)^2}\bigg)
	-\frac{d}{dz_3}\bigg(\frac{-8 z_1z_2z_3}{4z_3z_2(z_3-z_2)^24z_1z_3(z_1-z_3)^2}\bigg)
\end{align*}
which cancels exactly.

For $m=4$, $\mathcal{G}_{0,4}$ consists of 29 trees, which can be constructed via Lemma \ref{lem:Gnm1a} from the 4 trees of $\mathcal{G}_{0,3}$. Let $J=\{1,2,3,4\}$, it can be written as
\begin{align*}
	0=&W^{(0)}_{0,4}(y(z_J))\\
	=&W^{(0)}_{4,0}(x(z_J))
	-\sum_{i=1}^4\sum_{\substack{I_1\uplus I_2=J\setminus \{i\}\\ I_i\neq \emptyset}}\frac{1}{2!}\frac{d}{dz_i}\bigg(\prod_{k=1}^4(-2z_k)W^{(0)}_{|I_1|+1,0}(x(z_{I_1}),x(z_i))W^{(0)}_{|I_2|+1,0}(x(z_{I_2}),x(z_i))\bigg)\\
	+\sum_{\substack{i,j=1\\ i\neq j}}^4&\sum_{\substack{I_1\uplus I_2=J\setminus \{i,j\}\\ I_i\neq \emptyset}}\frac{1}{2!}\frac{d^2}{dz_idz_j}\bigg(\prod_{k=1}^4(-2z_k)W^{(0)}_{2,0}(x(z_i),x(z_j))W^{(0)}_{2,0}(x(z_{I_1}),x(z_i))W^{(0)}_{2,0}(x(z_{I_2}),x(z_j))\bigg)\\
	+\sum_{\substack{i=1}}^4&\sum_{\substack{I_1\uplus I_2\uplus I_3=J\setminus \{i\}\\ I_i\neq \emptyset}}\frac{1}{3!}\frac{d^2}{dz_i^2}\bigg(\prod_{k=1}^4(-2z_k)W^{(0)}_{2,0}(x(z_{I_1}),x(z_i))W^{(0)}_{2,0}(x(z_{I_2}),x(z_i))W^{(0)}_{2,0}(x(z_{I_3}),x(z_i))\bigg)
\end{align*}
which can be checked easily by computer algebra that it is true.

\section{Equivalence to other Results}\label{Sec:eq}
A functional relation between $W^{(g)}_{n,0}$ and $W^{(g)}_{0,m}$ was recently computed by the theory of double Hurwitz numbers, which is restricted to a certain class of spectral curves \cite[Theorem 3.4]{Borot:2021thu}. In case of $g=0$, the formula simplifies significantly, nevertheless it remains much more complicated than Theorem \ref{thm:first}. To avoid unnecessary definitions, we rephrase the results given on \cite[page 15]{Borot:2021thu} in our notation (we have switched the role of $x$ and $y$)
\begin{theorem}[\cite{Borot:2021thu}]
	The functional relation of $W^{(0)}_{0,m}$ reads
	\begin{align*}
		&W^{(0)}_{0,m}(y(z_1),...,y(z_m))\\
		=&\sum_{T\in \mathcal{G}_{0,m}} \frac{1}{\mathrm{Aut}(T)}\prod_{i=1}^m\frac{\sum_{k\geq 0}(-y(z_i)\frac{d}{dy(z_i)})^k (-\frac{y(z_i)dx(z_i)}{x(z_i)dy(z_i)})[v^k]\big(\partial_p+\frac{v}{p}\big)^{r_i(T)-1}\cdot 1\vert_{p=x(z_i)y(z_i)}x(z_i)^{r_i(T)}}{y(z_i)}\\
		&\qquad \times\prod_{(\emptyset,I)\in \mathcal{I}(T)} W^{(0)}_{|I|,0}(x(z_I)),
	\end{align*}
where the differential operator $(-y(z_i)\frac{d}{dy(z_i)})$ is acting on everything to its right and $[v^k]f(v)$ is the $k^{\text{th}}$ coefficient of $f(v)$.
\end{theorem}
The equivalence to Theorem \ref{thm:first} is realised by the following proposition
\begin{proposition}\label{pro:equiv}
	The following simplification holds
	\begin{align*}
		&\frac{1}{y(z)}\sum_{k\geq 0}\bigg(-y(z)\frac{d}{dy(z)}\bigg)^k \frac{y(z)}{x(z)}[v^k]\bigg(\partial_p+\frac{v}{p}\bigg)^{r-1}\cdot 1\vert_{p=x(z)y(z)}x(z)^{r}f(z)\\
		=&(-1)^{r-1}\frac{d^{r-1}}{dy(z)^{r-1}}f(z).
	\end{align*}
\end{proposition}
First, we observe
\begin{lemma}\label{lem:stirling}
	Let $s_{r,k}$ be the signed Stirling number defined by $s_{r+1,k}=s_{r,k-1}-r s_{r,k}$ with $s_{k,k}=1$, $s_{r,k}=0$ for $k>r$ and $s_{r,0}=\delta_{r,0}$. Then, the following holds
	\begin{align*}
		\big(\partial_p+\frac{v}{p}\big)^{r}\cdot 1=\frac{\sum_{k=0}^rs_{r,k}v^k}{p^r}
	\end{align*}
	\begin{proof}
		We prove by induction, where the initial case is trivially true,
		\begin{align*}
			&\big(\partial_p+\frac{v}{p}\big)^{r+1}\cdot 1=\big(\partial_p+\frac{v}{p}\big)\frac{\sum_{k=0}^rs_{r,k}v^k}{p^r}
			=\frac{\sum_{k=0}^r(-r\cdot s_{r,k})v^k}{p^{r+1}}+\frac{\sum_{k=0}^r s_{r,k}v^{k+1}}{p^{r+1}}\\
			=&\frac{\sum_{k=0}^{r+1}(s_{r,k-1}-rs_{r,k})v^k}{p^{r+1}}=\frac{\sum_{k=0}^{r+1}s_{r+1,k}v^k}{p^{r+1}}.
		\end{align*}
	\end{proof}
\end{lemma}
Let us finish the proof of Proposition \ref{pro:equiv}:
\begin{proof}[Proof of Proposition \ref{pro:equiv}]
	Rewrite the lhs of Proposition \ref{pro:equiv} by Lemma \ref{lem:stirling} to
	\begin{align*}
		&\frac{1}{y(z)}\sum_{k\geq 0}\bigg(-y(z)\frac{d}{dy(z)}\bigg)^k \frac{y(z)}{x(z)}[v^k]\bigg(\partial_p+\frac{v}{p}\bigg)^{r-1}\cdot 1\vert_{p=x(z_i)y(z_i)}x(z_i)^{r}f(z)\\
		=&\frac{1}{y(z)}\sum_{k= 0}^{r-1}\bigg(-y(z)\frac{d}{dy(z)}\bigg)^k \frac{y(z)}{x(z)}\frac{s_{r-1,k}}{x(z)^{r-1}y(z)^{r-1}}x(z)^rf(z)\\
		=&\frac{1}{y(z)}\sum_{k= 0}^{r-1}\bigg(-y(z)\frac{d}{dy(z)}\bigg)^k \frac{s_{r-1,k}f(z)}{y(z)^{r-2}}.
	\end{align*}
	Now, we proceed by induction. The initial case for $r=1$ gives
	\begin{align*}
		\frac{1}{y(z)}\frac{s_{0,0}f(z)}{y(z)^{-1}}=f(z).
	\end{align*}
	Now take $r\mapsto r+1$, which yields for $r>0$
	\begin{align*}
		&\frac{1}{y(z)}\sum_{k= 0}^{r}\bigg(-y(z)\frac{d}{dy(z)}\bigg)^k \frac{s_{r,k}f(z)}{y(z)^{r-1}}\\
		=&\frac{1}{y(z)}\sum_{k= 0}^{r}\bigg(-y(z)\frac{d}{dy(z)}\bigg)^k \frac{(s_{r-1,k-1}-(r-1) s_{r-1,k})f(z)}{y(z)^{r-1}}\\
		=&\frac{1}{y(z)}\bigg(-y(z)\frac{d}{dy(z)}\bigg)\frac{y(z)}{y(z)}\sum_{k=0}^{r-1}\bigg(-y(z)\frac{d}{dy(z)}\bigg)^k\frac{s_{r-1,k}\frac{f(z)}{y(z)}}{y(z)^{r-2}}\\
		&-(r-1)\frac{1}{y(z)}\sum_{k=0}^{r-1}\bigg(-y(z)\frac{d}{dy(z)}\bigg)^k\frac{s_{r-1,k}\frac{f(z)}{y(z)}}{y(z)^{r-2}}\\
		=&(-1)^r\frac{d}{dy(z)}\bigg[y(z)\frac{d^{r-1}}{dy(z)^{r-1}}+(r-1)\frac{d^{r-2}}{dy(z)^{r-2}}\bigg]\frac{f(z)}{y(z)},
	\end{align*}
	where we have used the recursive formula of $s_{r,k}$, the induction hypothesis for $r$ with $f(z)\mapsto \frac{f(z)}{y(z)}$. The proof is finished by the identity
	\begin{align*}
		\bigg[y(z)\frac{d^{r-1}}{dy(z)^{r-1}}+(r-1)\frac{d^{r-2}}{dy(z)^{r-2}}\bigg]\frac{f(z)}{y(z)}=\frac{d^{r-1}}{dy(z)^{r-1}} f(z).
	\end{align*}
	which is nothing else than the Leibniz rule of $\frac{d^{r-1}}{dy(z)^{r-1}} \big[\big(\frac{f(z)}{y(z)}\big) y(z)\big]$.
\end{proof}
Proposition \ref{pro:equiv} gives an enormous simplification to the results of \cite{Borot:2021thu} in case of $g=0$. Following their definitions of generating series of moments and free cumulants (see \cite{Voiculescu1986AdditionOC,Collins2006SecondOF} for more appropriate details), we define for a random variable $a$ 
\begin{align*}
	M_n(X_1,...,X_n):=&\sum_{k_1,...,k_n\geq 1} \varphi_n[a^{k_1},...,a^{k_n}]\prod_{i=1}^n X_i^{k_i}\\
	C_n(Y_1,...,Y_n):=&\sum_{k_1,...,k_n\geq 1} \kappa_{k_1,...,k_n}[a,...,a]\prod_{i=1}^n Y_i^{k_i}.
\end{align*}
The following functional relation holds \cite{Voiculescu1986AdditionOC,Collins2006SecondOF}
\begin{align*}
	C_1(X M_1(X))=M(X),
\end{align*}
\begin{align*}
	M_2(X_1,X_2)+\frac{X_1X_2}{(X_1-X_2)^2}=\frac{d \ln Y_1}{d \ln X_1}\frac{d \ln Y_2}{d \ln X_2}\bigg(C_2(Y_1,Y_2)+\frac{Y_1Y_2}{(Y_1-Y_2)^2}\bigg).
\end{align*}
We conclude as a corollary the following simplification of \cite[Theorem 1.1]{Borot:2021thu}
\begin{corollary}\label{cor:free}
	Consider  $Y_i=X_i M_1(X_i)$. For $n\geq 3$, we have
	\begin{align*}
		&M_n(X_1,...,X_n)\cdot X_1\cdot ...\cdot X_n\\
		=&\sum_{T\in \mathcal{G}_{0,n}}\frac{1}{\mathrm{Aut}(T)}\prod_{i=1}^n\bigg(X_i^2\frac{d}{dX_i}\bigg)^{r_i(T)-1}\bigg[\prod_{k=1}^n\bigg(X_k^2\frac{dY_k}{dX_k}\bigg)\prod_{(\emptyset,J)\in \mathcal{I}(T)}^\prime \frac{C_{|J|}(Y_J)}{\prod_{j\in J} Y_j}\bigg],
	\end{align*}
where $\mathcal{G}_{n,m}$ is the set of trees defined in Definition \ref{def:tree} and $r_j(T)$ is the valence of the $j^{\text{th}}$ $\circ$-vertex, $\mathcal{I}(T)$ the set of double sets corresponding to the $\bullet$-vertices in $T$ defined in Definition \ref{def:tree} and the primed product $\prod^\prime$ replaces $C_2(Y_i,Y_j)$ by $C_2(Y_i,Y_j)+\frac{Y_iY_j}{(Y_i-Y_j)^2}$.
\begin{proof}
	Identify $\frac{W^{(0)}_{n,0}(\frac{1}{x_1},...,\frac{1}{x_n})}{x_1\cdot ...\cdot x_n}=M_n(x_1,...,x_n)+\frac{\delta_{2,n}x_1x_2}{(x_1-x_2)^2}$ and $W^{(0)}_{0,m}(y_1,...,y_m)=\frac{C_m(y_1,...,y_m)}{y_1\cdot ...\cdot y_m}+\frac{\delta_{2,m}}{(y_1-y_2)^2}$ and apply Theorem \ref{thm:first} with $x$ and $y$ interchanged. 
	
	On the other hand, the corollary can also be proved by Theorem 1.1 of \cite{Borot:2021thu} and the application of Proposition \ref{pro:equiv}.
\end{proof}
\end{corollary}

\section{Higher Genus}
The benefit of the method used in this article is that it produces a very simple realisation of the functional relation between $W^{(g)}_{n,m}$. However, this is restricted to a fixed genus. For instance for genus $g=1$, we need the initial relation between $W^{(1)}_{1,0}$ and $W^{(1)}_{0,1}$. The functional relation for any genus are actually known \cite[Theorem 3.4]{Borot:2021thu}, but we expect that they should be simplified enormously by similar identities as in Proposition \ref{pro:equiv}.

We will underpin this assertion by the $g=1$ example. Recall $\omega^{(1)}_{1,0}$ is generated by TR \eqref{eq:trx} with $(\Sigma,x,y,B)$ and $\omega^{(1)}_{0,1}$ with $(\Sigma,y,x,B)$ as spectral curve. Then, the functional relation between those differential forms is well-known (see for instance \cite[Lemma C.1]{Eynard:2007kz})
\begin{align*}
	\omega^{(1)}_{1,0}(z)+\omega^{(1)}_{0,1}(z)=\frac{1}{24}d_z\bigg[\frac{1}{x'(z)y'(z)}\bigg(\frac{x''(z)y''(z)}{x'(z)y'(z)}+\frac{x''(z)^2}{x'(z)^2}-\frac{x'''(z)}{x'(z)}+\frac{y''(z)^2}{y'(z)^2}-\frac{y'''(z)}{y'(z)}\bigg)\bigg],
\end{align*}
where $d_zf(z)$ is the exterior derivative of some function $f(z)$, that is $f'(z)dz$.  Pulling this back to the $y$-plane by $y(z)=y$ and $x(z)=X(y)$ together with Fa\`a di Bruno's formula yields
\begin{align}\label{W11}
	W^{(1)}_{0,1}(y)=&-\frac{dX(y)}{dy}W^{(1)}_{1,0}(X(y))+\frac{1}{2}\frac{d}{dy}\bigg(\frac{dX(y)}{dy}\hat{W}^{(0)}_{2,0}(X(y),X(y))\bigg)
	-\frac{1}{24}\frac{d^3}{dy^3}\bigg(\frac{1}{\frac{dX(y)}{dy}}\bigg),
\end{align}
where $\hat{W}^{(0)}_{2,0}(x_1,x_2)=W^{(0)}_{2,0}(x_1,x_2)-\frac{1}{(x_1-x_2)^2}$ with a well-defined diagonal.

It is plausible that lower genus occur taken at the diagonal, for instance the second term on the rhs. However, there is even a further special term, the last term on the rhs of \eqref{W11}. Nevertheless, the loop insertion operator $\tilde{D}_{y_2}$ can be applied on the functional relation. By definition of the loop insertion operator, the result for $W^{(0)}_{2,1}$ on the diagonal and functional relation of $W^{(1)}_{1,1}$ (which is also easily computed), one can prove
\begin{proposition}\label{Prop:W111}
	The following functional relation holds
	\begin{align*}
		W^{(1)}_{0,2}(y_1,y_2)=&\frac{dX(y_1)}{dy_1}\frac{dX(y_2)}{dy_2}W^{(1)}_{2,0}(X(y_1),X(y_2))\\
		&-\frac{d}{dy_1}\bigg(\frac{dX(y_1)}{dy_1}\frac{dX(y_2)}{dy_2}W^{(1)}_{1,0}(X(y_1))W^{(0)}_{2,0}(X(y_1),X(y_2))\bigg)\\
		&-\frac{d}{dy_2}\bigg(\frac{dX(y_1)}{dy_1}\frac{dX(y_2)}{dy_2}W^{(1)}_{1,0}(X(y_2))W^{(0)}_{2,0}(X(y_1),X(y_2))\bigg)\\
		&-\frac{1}{2}\frac{d}{dy_1}\bigg(\frac{dX(y_1)}{dy_1}\frac{dX(y_2)}{dy_2}W^{(0)}_{3,0}(X(y_1),X(y_1),X(y_2))\bigg)\\
		&-\frac{1}{2}\frac{d}{dy_2}\bigg(\frac{dX(y_1)}{dy_1}\frac{dX(y_2)}{dy_2}W^{(0)}_{3,0}(X(y_2),X(y_2),X(y_1))\bigg)\\
		&+\frac{1}{2}\frac{d^2}{dy_1^2}\bigg(\frac{dX(y_1)}{dy_1}\frac{dX(y_2)}{dy_2}\hat{W}^{(0)}_{2,0}(X(y_1),X(y_1))W^{(0)}_{2,0}(X(y_1),X(y_2))\bigg)\\
		&+\frac{1}{2}\frac{d^2}{dy_2^2}\bigg(\frac{dX(y_1)}{dy_1}\frac{dX(y_2)}{dy_2}\hat{W}^{(0)}_{2,0}(X(y_2),X(y_2))W^{(0)}_{2,0}(X(y_1),X(y_2))\bigg)\\
		&+\frac{1}{2}\frac{d^2}{dy_1dy_2}\bigg(\frac{dX(y_1)}{dy_1}\frac{dX(y_2)}{dy_2}W^{(0)}_{2,0}(X(y_1),X(y_2))W^{(0)}_{2,0}(X(y_1),X(y_2))\bigg)\\
		&+\frac{1}{24}\frac{d^3}{dy_1^3}\bigg[\bigg(\frac{dy_1}{dX(y_1)}\bigg)^2\frac{d}{dy_1}\bigg(\frac{dX(y_1)}{dy_1}\frac{dX(y_2)}{dy_2}W^{(0)}_{2,0}(X(y_1),X(y_2))\bigg)\bigg]\\
		&+\frac{1}{24}\frac{d^3}{dy_2^3}\bigg[\bigg(\frac{dy_2}{dX(y_2)}\bigg)^2\frac{d}{dy_2}\bigg(\frac{dX(y_1)}{dy_1}\frac{dX(y_2)}{dy_2}W^{(0)}_{2,0}(X(y_1),X(y_2))\bigg)\bigg].
	\end{align*}
\end{proposition}
The last two lines of Proposition \ref{Prop:W111} give a new arising structure. All the other terms are in principle constructable by extending the definition of trees in Definition \ref{def:tree} to more general decorated graphs. To give results in $g>1$, the initial functional relation between  $W^{(g)}_{1,0}$ and $W^{(g)}_{0,1}$ has to be known, which is not achievable by the loop insertion operator.

\input{x-y-SymmetryV2.bbl}

\end{document}

%% file: xymap.pdf_tex
\begingroup%
  \makeatletter%
  \providecommand\color[2][]{%
    \errmessage{(Inkscape) Color is used for the text in Inkscape, but the package 'color.sty' is not loaded}%
    \renewcommand\color[2][]{}%
  }%
  \providecommand\transparent[1]{%
    \errmessage{(Inkscape) Transparency is used (non-zero) for the text in Inkscape, but the package 'transparent.sty' is not loaded}%
    \renewcommand\transparent[1]{}%
  }%
  \providecommand\rotatebox[2]{#2}%
  \newcommand*\fsize{\dimexpr\f@size pt\relax}%
  \newcommand*\lineheight[1]{\fontsize{\fsize}{#1\fsize}\selectfont}%
  \ifx\svgwidth\undefined%
    \setlength{\unitlength}{498.705134bp}%
    \ifx\svgscale\undefined%
      \relax%
    \else%
      \setlength{\unitlength}{\unitlength * \real{\svgscale}}%
    \fi%
  \else%
    \setlength{\unitlength}{\svgwidth}%
  \fi%
  \global\let\svgwidth\undefined%
  \global\let\svgscale\undefined%
  \makeatother%
  \begin{picture}(1,0.44457295)%
    \lineheight{1}%
    \setlength\tabcolsep{0pt}%
    \put(0,0){\includegraphics[width=\unitlength,page=1]{xymap.pdf}}%
    \put(0.48025088,0.38928325){\makebox(0,0)[lt]{\lineheight{1.25}\smash{\begin{tabular}[t]{l}$z$\end{tabular}}}}%
    \put(0.1139859,0.15064025){\makebox(0,0)[lt]{\lineheight{1.25}\smash{\begin{tabular}[t]{l}$x$\end{tabular}}}}%
    \put(0.83615434,0.15367857){\makebox(0,0)[lt]{\lineheight{1.25}\smash{\begin{tabular}[t]{l}$y$\end{tabular}}}}%
    \put(0.10843054,0.30079912){\makebox(0,0)[lt]{\lineheight{1.25}\smash{\begin{tabular}[t]{l}$x(z)$\end{tabular}}}}%
    \put(0.74728337,0.31167154){\makebox(0,0)[lt]{\lineheight{1.25}\smash{\begin{tabular}[t]{l}$y(z)$\end{tabular}}}}%
    \put(0.45260052,0.20735231){\makebox(0,0)[lt]{\lineheight{1.25}\smash{\begin{tabular}[t]{l}$Y(x)$\end{tabular}}}}%
    \put(0.46651133,0.00834192){\makebox(0,0)[lt]{\lineheight{1.25}\smash{\begin{tabular}[t]{l}$X(y)$\end{tabular}}}}%
  \end{picture}%
\endgroup%

%% file: G12.pdf_tex
\begingroup%
  \makeatletter%
  \providecommand\color[2][]{%
    \errmessage{(Inkscape) Color is used for the text in Inkscape, but the package 'color.sty' is not loaded}%
    \renewcommand\color[2][]{}%
  }%
  \providecommand\transparent[1]{%
    \errmessage{(Inkscape) Transparency is used (non-zero) for the text in Inkscape, but the package 'transparent.sty' is not loaded}%
    \renewcommand\transparent[1]{}%
  }%
  \providecommand\rotatebox[2]{#2}%
  \newcommand*\fsize{\dimexpr\f@size pt\relax}%
  \newcommand*\lineheight[1]{\fontsize{\fsize}{#1\fsize}\selectfont}%
  \ifx\svgwidth\undefined%
    \setlength{\unitlength}{219.29181719bp}%
    \ifx\svgscale\undefined%
      \relax%
    \else%
      \setlength{\unitlength}{\unitlength * \real{\svgscale}}%
    \fi%
  \else%
    \setlength{\unitlength}{\svgwidth}%
  \fi%
  \global\let\svgwidth\undefined%
  \global\let\svgscale\undefined%
  \makeatother%
  \begin{picture}(1,0.36711738)%
    \lineheight{1}%
    \setlength\tabcolsep{0pt}%
    \put(0.17547636,0.30450958){\makebox(0,0)[lt]{\lineheight{1.25}\smash{\begin{tabular}[t]{l}$1$\end{tabular}}}}%
    \put(0.03322153,0.04432138){\makebox(0,0)[lt]{\lineheight{1.25}\smash{\begin{tabular}[t]{l}$1$\end{tabular}}}}%
    \put(0.31922517,0.03961666){\makebox(0,0)[lt]{\lineheight{1.25}\smash{\begin{tabular}[t]{l}$2$\end{tabular}}}}%
    \put(0,0){\includegraphics[width=\unitlength,page=1]{G12.pdf}}%
    \put(0.75430227,0.30450955){\makebox(0,0)[lt]{\lineheight{1.25}\smash{\begin{tabular}[t]{l}$1$\end{tabular}}}}%
    \put(0.61118367,0.0400029){\makebox(0,0)[lt]{\lineheight{1.25}\smash{\begin{tabular}[t]{l}$1$\end{tabular}}}}%
    \put(0.90064224,0.03732191){\makebox(0,0)[lt]{\lineheight{1.25}\smash{\begin{tabular}[t]{l}$2$\end{tabular}}}}%
    \put(0,0){\includegraphics[width=\unitlength,page=2]{G12.pdf}}%
  \end{picture}%
\endgroup%

%% file: G21.pdf_tex
\begingroup%
  \makeatletter%
  \providecommand\color[2][]{%
    \errmessage{(Inkscape) Color is used for the text in Inkscape, but the package 'color.sty' is not loaded}%
    \renewcommand\color[2][]{}%
  }%
  \providecommand\transparent[1]{%
    \errmessage{(Inkscape) Transparency is used (non-zero) for the text in Inkscape, but the package 'transparent.sty' is not loaded}%
    \renewcommand\transparent[1]{}%
  }%
  \providecommand\rotatebox[2]{#2}%
  \newcommand*\fsize{\dimexpr\f@size pt\relax}%
  \newcommand*\lineheight[1]{\fontsize{\fsize}{#1\fsize}\selectfont}%
  \ifx\svgwidth\undefined%
    \setlength{\unitlength}{332.03785648bp}%
    \ifx\svgscale\undefined%
      \relax%
    \else%
      \setlength{\unitlength}{\unitlength * \real{\svgscale}}%
    \fi%
  \else%
    \setlength{\unitlength}{\svgwidth}%
  \fi%
  \global\let\svgwidth\undefined%
  \global\let\svgscale\undefined%
  \makeatother%
  \begin{picture}(1,0.25375369)%
    \lineheight{1}%
    \setlength\tabcolsep{0pt}%
    \put(0.11418071,0.19939969){\makebox(0,0)[lt]{\lineheight{1.25}\smash{\begin{tabular}[t]{l}$1$\end{tabular}}}}%
    \put(0.01737748,0.02527873){\makebox(0,0)[lt]{\lineheight{1.25}\smash{\begin{tabular}[t]{l}$1$\end{tabular}}}}%
    \put(0.20484028,0.04185124){\makebox(0,0)[lt]{\lineheight{1.25}\smash{\begin{tabular}[t]{l}$2$\end{tabular}}}}%
    \put(0,0){\includegraphics[width=\unitlength,page=1]{G21.pdf}}%
    \put(0.49788783,0.20025531){\makebox(0,0)[lt]{\lineheight{1.25}\smash{\begin{tabular}[t]{l}$1$\end{tabular}}}}%
    \put(0.40022894,0.02499353){\makebox(0,0)[lt]{\lineheight{1.25}\smash{\begin{tabular}[t]{l}$1$\end{tabular}}}}%
    \put(0.58341364,0.03691315){\makebox(0,0)[lt]{\lineheight{1.25}\smash{\begin{tabular}[t]{l}$2$\end{tabular}}}}%
    \put(0,0){\includegraphics[width=\unitlength,page=2]{G21.pdf}}%
    \put(0.84856895,0.21211965){\makebox(0,0)[lt]{\lineheight{1.25}\smash{\begin{tabular}[t]{l}$1$\end{tabular}}}}%
    \put(0.75176577,0.03913953){\makebox(0,0)[lt]{\lineheight{1.25}\smash{\begin{tabular}[t]{l}$1$\end{tabular}}}}%
    \put(0.93437994,0.0470662){\makebox(0,0)[lt]{\lineheight{1.25}\smash{\begin{tabular}[t]{l}$2$\end{tabular}}}}%
    \put(0,0){\includegraphics[width=\unitlength,page=3]{G21.pdf}}%
  \end{picture}%
\endgroup%

%% file: G31.pdf_tex
\begingroup%
  \makeatletter%
  \providecommand\color[2][]{%
    \errmessage{(Inkscape) Color is used for the text in Inkscape, but the package 'color.sty' is not loaded}%
    \renewcommand\color[2][]{}%
  }%
  \providecommand\transparent[1]{%
    \errmessage{(Inkscape) Transparency is used (non-zero) for the text in Inkscape, but the package 'transparent.sty' is not loaded}%
    \renewcommand\transparent[1]{}%
  }%
  \providecommand\rotatebox[2]{#2}%
  \newcommand*\fsize{\dimexpr\f@size pt\relax}%
  \newcommand*\lineheight[1]{\fontsize{\fsize}{#1\fsize}\selectfont}%
  \ifx\svgwidth\undefined%
    \setlength{\unitlength}{343.25095617bp}%
    \ifx\svgscale\undefined%
      \relax%
    \else%
      \setlength{\unitlength}{\unitlength * \real{\svgscale}}%
    \fi%
  \else%
    \setlength{\unitlength}{\svgwidth}%
  \fi%
  \global\let\svgwidth\undefined%
  \global\let\svgscale\undefined%
  \makeatother%
  \begin{picture}(1,0.23453926)%
    \lineheight{1}%
    \setlength\tabcolsep{0pt}%
    \put(0.11210611,0.19454118){\makebox(0,0)[lt]{\lineheight{1.25}\smash{\begin{tabular}[t]{l}$1$\end{tabular}}}}%
    \put(0.02122415,0.02831545){\makebox(0,0)[lt]{\lineheight{1.25}\smash{\begin{tabular}[t]{l}$1$\end{tabular}}}}%
    \put(0.20394252,0.02530976){\makebox(0,0)[lt]{\lineheight{1.25}\smash{\begin{tabular}[t]{l}$3$\end{tabular}}}}%
    \put(0,0){\includegraphics[width=\unitlength,page=1]{G31.pdf}}%
    \put(0.48189907,0.19454117){\makebox(0,0)[lt]{\lineheight{1.25}\smash{\begin{tabular}[t]{l}$1$\end{tabular}}}}%
    \put(0,0){\includegraphics[width=\unitlength,page=2]{G31.pdf}}%
    \put(0.11299371,0.02831545){\makebox(0,0)[lt]{\lineheight{1.25}\smash{\begin{tabular}[t]{l}$2$\end{tabular}}}}%
    \put(0,0){\includegraphics[width=\unitlength,page=3]{G31.pdf}}%
    \put(0.39328182,0.02969738){\makebox(0,0)[lt]{\lineheight{1.25}\smash{\begin{tabular}[t]{l}$i$\end{tabular}}}}%
    \put(0.5760002,0.02669169){\makebox(0,0)[lt]{\lineheight{1.25}\smash{\begin{tabular}[t]{l}$k$\end{tabular}}}}%
    \put(0,0){\includegraphics[width=\unitlength,page=4]{G31.pdf}}%
    \put(0.48505135,0.02969738){\makebox(0,0)[lt]{\lineheight{1.25}\smash{\begin{tabular}[t]{l}$j$\end{tabular}}}}%
    \put(0,0){\includegraphics[width=\unitlength,page=5]{G31.pdf}}%
    \put(0.84242245,0.19454117){\makebox(0,0)[lt]{\lineheight{1.25}\smash{\begin{tabular}[t]{l}$1$\end{tabular}}}}%
    \put(0,0){\includegraphics[width=\unitlength,page=6]{G31.pdf}}%
    \put(0.75380526,0.02969741){\makebox(0,0)[lt]{\lineheight{1.25}\smash{\begin{tabular}[t]{l}$1$\end{tabular}}}}%
    \put(0.93652358,0.02669172){\makebox(0,0)[lt]{\lineheight{1.25}\smash{\begin{tabular}[t]{l}$3$\end{tabular}}}}%
    \put(0,0){\includegraphics[width=\unitlength,page=7]{G31.pdf}}%
    \put(0.84557467,0.02969741){\makebox(0,0)[lt]{\lineheight{1.25}\smash{\begin{tabular}[t]{l}$2$\end{tabular}}}}%
    \put(0,0){\includegraphics[width=\unitlength,page=8]{G31.pdf}}%
  \end{picture}%
\endgroup%

%% file: Gn1m.pdf_tex
\begingroup%
  \makeatletter%
  \providecommand\color[2][]{%
    \errmessage{(Inkscape) Color is used for the text in Inkscape, but the package 'color.sty' is not loaded}%
    \renewcommand\color[2][]{}%
  }%
  \providecommand\transparent[1]{%
    \errmessage{(Inkscape) Transparency is used (non-zero) for the text in Inkscape, but the package 'transparent.sty' is not loaded}%
    \renewcommand\transparent[1]{}%
  }%
  \providecommand\rotatebox[2]{#2}%
  \newcommand*\fsize{\dimexpr\f@size pt\relax}%
  \newcommand*\lineheight[1]{\fontsize{\fsize}{#1\fsize}\selectfont}%
  \ifx\svgwidth\undefined%
    \setlength{\unitlength}{357.01414656bp}%
    \ifx\svgscale\undefined%
      \relax%
    \else%
      \setlength{\unitlength}{\unitlength * \real{\svgscale}}%
    \fi%
  \else%
    \setlength{\unitlength}{\svgwidth}%
  \fi%
  \global\let\svgwidth\undefined%
  \global\let\svgscale\undefined%
  \makeatother%
  \begin{picture}(1,0.15904225)%
    \lineheight{1}%
    \setlength\tabcolsep{0pt}%
    \put(0,0){\includegraphics[width=\unitlength,page=1]{Gn1m.pdf}}%
    \put(0.27958804,0.07448303){\makebox(0,0)[lt]{\lineheight{1.25}\smash{\begin{tabular}[t]{l}$n+1$\end{tabular}}}}%
    \put(0,0){\includegraphics[width=\unitlength,page=2]{Gn1m.pdf}}%
    \put(0.09118219,0.0746931){\makebox(0,0)[lt]{\lineheight{1.25}\smash{\begin{tabular}[t]{l}$k$\end{tabular}}}}%
    \put(0,0){\includegraphics[width=\unitlength,page=3]{Gn1m.pdf}}%
    \put(0.8152811,0.07448303){\makebox(0,0)[lt]{\lineheight{1.25}\smash{\begin{tabular}[t]{l}$n+1$\end{tabular}}}}%
    \put(0,0){\includegraphics[width=\unitlength,page=4]{Gn1m.pdf}}%
  \end{picture}%
\endgroup%

%% file: Gnm1.pdf_tex
\begingroup%
  \makeatletter%
  \providecommand\color[2][]{%
    \errmessage{(Inkscape) Color is used for the text in Inkscape, but the package 'color.sty' is not loaded}%
    \renewcommand\color[2][]{}%
  }%
  \providecommand\transparent[1]{%
    \errmessage{(Inkscape) Transparency is used (non-zero) for the text in Inkscape, but the package 'transparent.sty' is not loaded}%
    \renewcommand\transparent[1]{}%
  }%
  \providecommand\rotatebox[2]{#2}%
  \newcommand*\fsize{\dimexpr\f@size pt\relax}%
  \newcommand*\lineheight[1]{\fontsize{\fsize}{#1\fsize}\selectfont}%
  \ifx\svgwidth\undefined%
    \setlength{\unitlength}{490.96498917bp}%
    \ifx\svgscale\undefined%
      \relax%
    \else%
      \setlength{\unitlength}{\unitlength * \real{\svgscale}}%
    \fi%
  \else%
    \setlength{\unitlength}{\svgwidth}%
  \fi%
  \global\let\svgwidth\undefined%
  \global\let\svgscale\undefined%
  \makeatother%
  \begin{picture}(1,0.44300756)%
    \lineheight{1}%
    \setlength\tabcolsep{0pt}%
    \put(0.08479399,0.1679399){\makebox(0,0)[lt]{\lineheight{1.25}\smash{\begin{tabular}[t]{l}$m+1$\end{tabular}}}}%
    \put(0,0){\includegraphics[width=\unitlength,page=1]{Gnm1.pdf}}%
    \put(0.01475603,0.01621862){\makebox(0,0)[lt]{\lineheight{1.25}\smash{\begin{tabular}[t]{l}$T_1$\end{tabular}}}}%
    \put(0.0764833,0.01668197){\makebox(0,0)[lt]{\lineheight{1.25}\smash{\begin{tabular}[t]{l}$T_2$\end{tabular}}}}%
    \put(0.196758,0.01689657){\makebox(0,0)[lt]{\lineheight{1.25}\smash{\begin{tabular}[t]{l}$T_k$\end{tabular}}}}%
    \put(0,0){\includegraphics[width=\unitlength,page=2]{Gnm1.pdf}}%
    \put(0.39031475,0.16793987){\makebox(0,0)[lt]{\lineheight{1.25}\smash{\begin{tabular}[t]{l}$m+1$\end{tabular}}}}%
    \put(0,0){\includegraphics[width=\unitlength,page=3]{Gnm1.pdf}}%
    \put(0.32027682,0.01621862){\makebox(0,0)[lt]{\lineheight{1.25}\smash{\begin{tabular}[t]{l}$T_1$\end{tabular}}}}%
    \put(0.38200405,0.01668192){\makebox(0,0)[lt]{\lineheight{1.25}\smash{\begin{tabular}[t]{l}$T_2$\end{tabular}}}}%
    \put(0.50227886,0.01689657){\makebox(0,0)[lt]{\lineheight{1.25}\smash{\begin{tabular}[t]{l}$T_k$\end{tabular}}}}%
    \put(0,0){\includegraphics[width=\unitlength,page=4]{Gnm1.pdf}}%
    \put(0.81804386,0.1679399){\makebox(0,0)[lt]{\lineheight{1.25}\smash{\begin{tabular}[t]{l}$m+1$\end{tabular}}}}%
    \put(0,0){\includegraphics[width=\unitlength,page=5]{Gnm1.pdf}}%
    \put(0.74800585,0.01621862){\makebox(0,0)[lt]{\lineheight{1.25}\smash{\begin{tabular}[t]{l}$T_1$\end{tabular}}}}%
    \put(0.80973316,0.01668197){\makebox(0,0)[lt]{\lineheight{1.25}\smash{\begin{tabular}[t]{l}$T_2$\end{tabular}}}}%
    \put(0.93000785,0.01689657){\makebox(0,0)[lt]{\lineheight{1.25}\smash{\begin{tabular}[t]{l}$T_k$\end{tabular}}}}%
    \put(0,0){\includegraphics[width=\unitlength,page=6]{Gnm1.pdf}}%
    \put(0.33555283,0.36603991){\makebox(0,0)[lt]{\lineheight{1.25}\smash{\begin{tabular}[t]{l}$T_1$\end{tabular}}}}%
    \put(0.39728011,0.36650326){\makebox(0,0)[lt]{\lineheight{1.25}\smash{\begin{tabular}[t]{l}$T_2$\end{tabular}}}}%
    \put(0.51755479,0.36671784){\makebox(0,0)[lt]{\lineheight{1.25}\smash{\begin{tabular}[t]{l}$T_k$\end{tabular}}}}%
    \put(0,0){\includegraphics[width=\unitlength,page=7]{Gnm1.pdf}}%
  \end{picture}%
\endgroup%

%% file: DxcircWn1.pdf_tex
\begingroup%
  \makeatletter%
  \providecommand\color[2][]{%
    \errmessage{(Inkscape) Color is used for the text in Inkscape, but the package 'color.sty' is not loaded}%
    \renewcommand\color[2][]{}%
  }%
  \providecommand\transparent[1]{%
    \errmessage{(Inkscape) Transparency is used (non-zero) for the text in Inkscape, but the package 'transparent.sty' is not loaded}%
    \renewcommand\transparent[1]{}%
  }%
  \providecommand\rotatebox[2]{#2}%
  \newcommand*\fsize{\dimexpr\f@size pt\relax}%
  \newcommand*\lineheight[1]{\fontsize{\fsize}{#1\fsize}\selectfont}%
  \ifx\svgwidth\undefined%
    \setlength{\unitlength}{307.36750642bp}%
    \ifx\svgscale\undefined%
      \relax%
    \else%
      \setlength{\unitlength}{\unitlength * \real{\svgscale}}%
    \fi%
  \else%
    \setlength{\unitlength}{\svgwidth}%
  \fi%
  \global\let\svgwidth\undefined%
  \global\let\svgscale\undefined%
  \makeatother%
  \begin{picture}(1,0.13531367)%
    \lineheight{1}%
    \setlength\tabcolsep{0pt}%
    \put(0,0){\includegraphics[width=\unitlength,page=1]{DxcircWn1.pdf}}%
    \put(0.20069077,0.0685921){\makebox(0,0)[lt]{\lineheight{1.25}\smash{\begin{tabular}[t]{l}$y$\end{tabular}}}}%
    \put(0,0){\includegraphics[width=\unitlength,page=2]{DxcircWn1.pdf}}%
    \put(-0.00540121,0.06570345){\makebox(0,0)[lt]{\lineheight{1.25}\smash{\begin{tabular}[t]{l}$D_{x_{n+1}}$\end{tabular}}}}%
    \put(0,0){\includegraphics[width=\unitlength,page=3]{DxcircWn1.pdf}}%
    \put(0.40793019,0.06496271){\makebox(0,0)[lt]{\lineheight{1.25}\smash{\begin{tabular}[t]{l}$=$\end{tabular}}}}%
    \put(0,0){\includegraphics[width=\unitlength,page=4]{DxcircWn1.pdf}}%
    \put(0.61242252,0.07167313){\makebox(0,0)[lt]{\lineheight{1.25}\smash{\begin{tabular}[t]{l}$y$\end{tabular}}}}%
    \put(0,0){\includegraphics[width=\unitlength,page=5]{DxcircWn1.pdf}}%
    \put(0.78544492,0.06501867){\makebox(0,0)[lt]{\lineheight{1.25}\smash{\begin{tabular}[t]{l}$n+1$\end{tabular}}}}%
  \end{picture}%
\endgroup%

%% file: Dxbullet.pdf_tex
\begingroup%
  \makeatletter%
  \providecommand\color[2][]{%
    \errmessage{(Inkscape) Color is used for the text in Inkscape, but the package 'color.sty' is not loaded}%
    \renewcommand\color[2][]{}%
  }%
  \providecommand\transparent[1]{%
    \errmessage{(Inkscape) Transparency is used (non-zero) for the text in Inkscape, but the package 'transparent.sty' is not loaded}%
    \renewcommand\transparent[1]{}%
  }%
  \providecommand\rotatebox[2]{#2}%
  \newcommand*\fsize{\dimexpr\f@size pt\relax}%
  \newcommand*\lineheight[1]{\fontsize{\fsize}{#1\fsize}\selectfont}%
  \ifx\svgwidth\undefined%
    \setlength{\unitlength}{273.61750642bp}%
    \ifx\svgscale\undefined%
      \relax%
    \else%
      \setlength{\unitlength}{\unitlength * \real{\svgscale}}%
    \fi%
  \else%
    \setlength{\unitlength}{\svgwidth}%
  \fi%
  \global\let\svgwidth\undefined%
  \global\let\svgscale\undefined%
  \makeatother%
  \begin{picture}(1,0.14963669)%
    \lineheight{1}%
    \setlength\tabcolsep{0pt}%
    \put(-0.00606743,0.07144024){\makebox(0,0)[lt]{\lineheight{1.25}\smash{\begin{tabular}[t]{l}$D_{x_{n+1}}$\end{tabular}}}}%
    \put(0,0){\includegraphics[width=\unitlength,page=1]{Dxbullet.pdf}}%
    \put(0.45824731,0.07060813){\makebox(0,0)[lt]{\lineheight{1.25}\smash{\begin{tabular}[t]{l}$=$\end{tabular}}}}%
    \put(0,0){\includegraphics[width=\unitlength,page=2]{Dxbullet.pdf}}%
    \put(0.75898011,0.07067098){\makebox(0,0)[lt]{\lineheight{1.25}\smash{\begin{tabular}[t]{l}$n+1$\end{tabular}}}}%
    \put(0,0){\includegraphics[width=\unitlength,page=3]{Dxbullet.pdf}}%
  \end{picture}%
\endgroup%

%% file: Dxcirc.pdf_tex
\begingroup%
  \makeatletter%
  \providecommand\color[2][]{%
    \errmessage{(Inkscape) Color is used for the text in Inkscape, but the package 'color.sty' is not loaded}%
    \renewcommand\color[2][]{}%
  }%
  \providecommand\transparent[1]{%
    \errmessage{(Inkscape) Transparency is used (non-zero) for the text in Inkscape, but the package 'transparent.sty' is not loaded}%
    \renewcommand\transparent[1]{}%
  }%
  \providecommand\rotatebox[2]{#2}%
  \newcommand*\fsize{\dimexpr\f@size pt\relax}%
  \newcommand*\lineheight[1]{\fontsize{\fsize}{#1\fsize}\selectfont}%
  \ifx\svgwidth\undefined%
    \setlength{\unitlength}{307.36750642bp}%
    \ifx\svgscale\undefined%
      \relax%
    \else%
      \setlength{\unitlength}{\unitlength * \real{\svgscale}}%
    \fi%
  \else%
    \setlength{\unitlength}{\svgwidth}%
  \fi%
  \global\let\svgwidth\undefined%
  \global\let\svgscale\undefined%
  \makeatother%
  \begin{picture}(1,0.13531367)%
    \lineheight{1}%
    \setlength\tabcolsep{0pt}%
    \put(0,0){\includegraphics[width=\unitlength,page=1]{Dxcirc.pdf}}%
    \put(0.20069077,0.0685921){\makebox(0,0)[lt]{\lineheight{1.25}\smash{\begin{tabular}[t]{l}$j$\end{tabular}}}}%
    \put(0,0){\includegraphics[width=\unitlength,page=2]{Dxcirc.pdf}}%
    \put(-0.00540121,0.06570345){\makebox(0,0)[lt]{\lineheight{1.25}\smash{\begin{tabular}[t]{l}$D_{x_{n+1}}$\end{tabular}}}}%
    \put(0,0){\includegraphics[width=\unitlength,page=3]{Dxcirc.pdf}}%
    \put(0.40793019,0.06496271){\makebox(0,0)[lt]{\lineheight{1.25}\smash{\begin{tabular}[t]{l}$=$\end{tabular}}}}%
    \put(0,0){\includegraphics[width=\unitlength,page=4]{Dxcirc.pdf}}%
    \put(0.61242252,0.07167313){\makebox(0,0)[lt]{\lineheight{1.25}\smash{\begin{tabular}[t]{l}$j$\end{tabular}}}}%
    \put(0,0){\includegraphics[width=\unitlength,page=5]{Dxcirc.pdf}}%
    \put(0.78544492,0.06501867){\makebox(0,0)[lt]{\lineheight{1.25}\smash{\begin{tabular}[t]{l}$n+1$\end{tabular}}}}%
  \end{picture}%
\endgroup%

%% file: Dycirc.pdf_tex
\begingroup%
  \makeatletter%
  \providecommand\color[2][]{%
    \errmessage{(Inkscape) Color is used for the text in Inkscape, but the package 'color.sty' is not loaded}%
    \renewcommand\color[2][]{}%
  }%
  \providecommand\transparent[1]{%
    \errmessage{(Inkscape) Transparency is used (non-zero) for the text in Inkscape, but the package 'transparent.sty' is not loaded}%
    \renewcommand\transparent[1]{}%
  }%
  \providecommand\rotatebox[2]{#2}%
  \newcommand*\fsize{\dimexpr\f@size pt\relax}%
  \newcommand*\lineheight[1]{\fontsize{\fsize}{#1\fsize}\selectfont}%
  \ifx\svgwidth\undefined%
    \setlength{\unitlength}{314.17415305bp}%
    \ifx\svgscale\undefined%
      \relax%
    \else%
      \setlength{\unitlength}{\unitlength * \real{\svgscale}}%
    \fi%
  \else%
    \setlength{\unitlength}{\svgwidth}%
  \fi%
  \global\let\svgwidth\undefined%
  \global\let\svgscale\undefined%
  \makeatother%
  \begin{picture}(1,0.13238207)%
    \lineheight{1}%
    \setlength\tabcolsep{0pt}%
    \put(0,0){\includegraphics[width=\unitlength,page=1]{Dycirc.pdf}}%
    \put(0.19634276,0.06710604){\makebox(0,0)[lt]{\lineheight{1.25}\smash{\begin{tabular}[t]{l}$j$\end{tabular}}}}%
    \put(0,0){\includegraphics[width=\unitlength,page=2]{Dycirc.pdf}}%
    \put(-0.00528419,0.06427998){\makebox(0,0)[lt]{\lineheight{1.25}\smash{\begin{tabular}[t]{l}$\tilde{D}_{y_{m+1}}$\end{tabular}}}}%
    \put(0,0){\includegraphics[width=\unitlength,page=3]{Dycirc.pdf}}%
    \put(0.3990923,0.06355528){\makebox(0,0)[lt]{\lineheight{1.25}\smash{\begin{tabular}[t]{l}$=$\end{tabular}}}}%
    \put(0,0){\includegraphics[width=\unitlength,page=4]{Dycirc.pdf}}%
    \put(0.59915426,0.07012032){\makebox(0,0)[lt]{\lineheight{1.25}\smash{\begin{tabular}[t]{l}$j$\end{tabular}}}}%
    \put(0,0){\includegraphics[width=\unitlength,page=5]{Dycirc.pdf}}%
    \put(0.7684281,0.06361002){\makebox(0,0)[lt]{\lineheight{1.25}\smash{\begin{tabular}[t]{l}$m+1$\end{tabular}}}}%
    \put(0,0){\includegraphics[width=\unitlength,page=6]{Dycirc.pdf}}%
  \end{picture}%
\endgroup%

%% file: Gnm1Proof.pdf_tex
\begingroup%
  \makeatletter%
  \providecommand\color[2][]{%
    \errmessage{(Inkscape) Color is used for the text in Inkscape, but the package 'color.sty' is not loaded}%
    \renewcommand\color[2][]{}%
  }%
  \providecommand\transparent[1]{%
    \errmessage{(Inkscape) Transparency is used (non-zero) for the text in Inkscape, but the package 'transparent.sty' is not loaded}%
    \renewcommand\transparent[1]{}%
  }%
  \providecommand\rotatebox[2]{#2}%
  \newcommand*\fsize{\dimexpr\f@size pt\relax}%
  \newcommand*\lineheight[1]{\fontsize{\fsize}{#1\fsize}\selectfont}%
  \ifx\svgwidth\undefined%
    \setlength{\unitlength}{435.51812327bp}%
    \ifx\svgscale\undefined%
      \relax%
    \else%
      \setlength{\unitlength}{\unitlength * \real{\svgscale}}%
    \fi%
  \else%
    \setlength{\unitlength}{\svgwidth}%
  \fi%
  \global\let\svgwidth\undefined%
  \global\let\svgscale\undefined%
  \makeatother%
  \begin{picture}(1,0.48345673)%
    \lineheight{1}%
    \setlength\tabcolsep{0pt}%
    \put(0.03491024,0.13931257){\makebox(0,0)[lt]{\lineheight{1.25}\smash{\begin{tabular}[t]{l}$m+1$\end{tabular}}}}%
    \put(0,0){\includegraphics[width=\unitlength,page=1]{Gnm1Proof.pdf}}%
    \put(0.3793276,0.13931254){\makebox(0,0)[lt]{\lineheight{1.25}\smash{\begin{tabular}[t]{l}$m+1$\end{tabular}}}}%
    \put(0,0){\includegraphics[width=\unitlength,page=2]{Gnm1Proof.pdf}}%
    \put(0.86151193,0.13931257){\makebox(0,0)[lt]{\lineheight{1.25}\smash{\begin{tabular}[t]{l}$m+1$\end{tabular}}}}%
    \put(0,0){\includegraphics[width=\unitlength,page=3]{Gnm1Proof.pdf}}%
    \put(0.18030632,0.42370948){\makebox(0,0)[lt]{\lineheight{1.25}\smash{\begin{tabular}[t]{l}$\tilde{D}_{y_{m+1}}$\end{tabular}}}}%
    \put(0,0){\includegraphics[width=\unitlength,page=4]{Gnm1Proof.pdf}}%
  \end{picture}%
\endgroup%

%% file: x-y-SymmetryV2.bbl
\newcommand{\etalchar}[1]{$^{#1}$}